\documentclass[a4paper,11pt]{article}

\usepackage{preamble-fv-simple}
\usepackage{tikz}
\usetikzlibrary{graphs,graphs.standard}
\usetikzlibrary{calc}
\usetikzlibrary{decorations.pathreplacing}

\author[1]{Hans L.\ Bodlaender}
\author[2]{Lars Jaffke} 
\author[2]{Jan Arne Telle}

\affil[1]{Utrecht University, The Netherlands}
\affil[ ]{\texttt{h.l.bodlaender@uu.nl}}
\affil[2]{University of Bergen, Norway} 
\affil[ ]{\texttt{\{lars.jaffke,jan.arne.telle\}@uib.no}} 

\title{Typical Sequences Revisited --- Computing Width Parameters of Graphs\thanks{This work was started when the third author was visiting Universitat Politecnica de Valencia, and part of it was done while the second author was visiting Utrecht University. The first author was partially supported by the Networks project, funded by the Netherlands Organization for Scientific Research (NWO). The second author is supported by the Bergen Research Foundation (BFS).}} 

\usepackage{macros}
\newcommand\calO{\mathcal{O}}

\begin{document}

\maketitle

\begin{abstract}
In this work, we give a structural lemma on merges of typical sequences, a notion that was 
introduced in 1991 [Lagergren and Arnborg, Bodlaender and Kloks, both ICALP 1991]
to obtain constructive linear time parameterized algorithms for treewidth and pathwidth.
The lemma addresses a runtime bottleneck in those algorithms but so far it does not lead to asymptotically faster algorithms. 
However, we apply the lemma to show that the cutwidth and the modified cutwidth of series parallel digraphs can be computed in $\cO(n^2)$ time. 
\end{abstract}

\section{Introduction}

In this paper we revisit an old key technique from what currently are the theoretically fastest parameterized algorithms for treewidth and pathwidth, namely the use of {\em typical sequences}, and give additional
structural insights for this technique. In particular, we show a structural lemma, which we call the \emph{Merge Dominator Lemma}. The technique of typical sequences brings with it a partial ordering on sequences of integers, and a notion of possible merges of two integer sequences; 
surprisingly, the Merge Dominator Lemma states that for any pair of integer sequences there exists
a \emph{single} merge that dominates all merges of these integer sequences, and this dominating merge can be found in linear time.
On its own, this lemma does not lead to asymptotically faster parameterized algorithms for treewidth and pathwidth, 
but, as we discuss below, it is a concrete step towards such algorithms.

The notion of typical sequences was introduced independently in 1991 by Lagergren and Arnborg \cite{LagergrenA91} and Bodlaender and Kloks \cite{BK96}. In both papers, it is a key element in an explicit dynamic programming algorithm that given a tree decomposition of bounded width $\ell$, decides if the pathwidth or treewidth of the input graph $G$ is at most a constant $k$. Lagergren and Arnborg build 
upon this result and show that the set of forbidden minors of graphs of treewidth (or pathwidth) at most
$k$ is computable; Bodlaender and Kloks show that the algorithm can also construct a tree or path decomposition of width at most $k$, if existing, in the same asymptotic time bounds. The latter result is
a main subroutine in Bodlaender's linear time algorithm \cite{Bodlaender96}
for treewidth-$k$. If one analyses the running time
of Bodlaender's algorithm for treewidth or pathwidth $\leq k$, then one can observe that the bottleneck is in the subroutine that calls the Bodlaender-Kloks dynamic programming subroutine, with both the subroutine and the main algorithm having time $\cO(2^{\cO(k^3)}n)$ for treewidth, and $\cO(2^{\cO(k^2)}n)$ for pathwidth.
See also the recent work by F\"{u}rer for pathwidth \cite{Furer16}, 
and the simplified versions of the algorithms of~\cite{Bodlaender96,BK96} by Althaus and Ziegler~\cite{AZ19}.  
Now, over a quarter of a century after the discovery of these results, even though much work has been done on treewidth recognition algorithms
(see e.g.~\cite{Ami10,Bod16,BP17,FHL08,Furer16,Lag96,Ree92,RS95}), 
these bounds on the function of $k$ are still the best known, i.e.\
no $\cO(2^{o(k^3)}n^{O(1)})$ algorithm for treewidth, and no $\cO(2^{o(k^2)}n^{O(1)})$ algorithm
for pathwidth is known. An interesting question, and a long-standing open problem in the field \cite[Problem 2.7.1]{IPEC2006}, is whether such algorithms can be obtained. Possible approaches
to answer such a question is to design (e.g.\ ETH or SETH based) lower bounds, find an entirely new approach
to compute treewidth or pathwidth in a parameterized setting, or improve upon the dynamic programming 
algorithms of \cite{LagergrenA91} and \cite{BK96}. 
Using our Merge Dominator Lemma we can go one step towards the latter, as follows.

The algorithms of Lagergren and Arnborg \cite{LagergrenA91} and Bodlaender and Kloks \cite{BK96} are based upon tabulating characteristics of tree or path decompositions of subgraphs of the input graph; 
a characteristic consists of an {\em intersection model}, that tells how the vertices in the current top bag
interact, and for each {\em part} of the intersection model, a typical sequence of bag sizes.\footnote{This approach was later used in several follow up results to obtain explicit constructive parameterized algorithms for
other graph width measures, like cutwidth \cite{ThilikosSB05,ThilikosSB05a}, branchwidth \cite{BodlaenderT97}, different types of search numbers like linear width \cite{BodlaenderT04}, and
directed vertex separation number \cite{BodlaenderGT98}.}
The set of characteristics for a join node is computed from the sets of characteristics of its (two) children. 
In particular, each pair of characteristics with one from each child can give rise to exponentially (in $k$) many characteristics for the join node.
This is because exponentially many typical sequences may arise as the merges of the typical sequences that are part of the characteristics.
In the light of our Merge Dominator Lemma, only \emph{one} of these merges has to be stored,
reducing the number of characteristics arising from each pair of characteristics of the children
from~$2^{\cO(k)}$ to just $1$.
Moreover, this dominating merge can be found in $\cO(k)$ time, with no large constants hidden in the `$\cO$'. 

Merging typical sequences at a join node is however not the only way the number of characteristics can increase throughout the algorithm, e.g.\ at introduce nodes, the number of characteristics increases in a different way.
Nevertheless, the number of intersection models is $\cO(k^{\cO(k)})$ for pathwidth and $\cO(k^{\cO(k^2)})$ for treewidth; perhaps, with additional techniques, the number of typical sequences per part can be better bounded --- 
in the case that a single dominating typical sequence per part suffices, this would
reduce the number of table entries per node to $\cO(k^{\cO(k)})$ for pathwidth-$k$, and to $\cO(k^{\cO(k^2)})$ for treewidth-$k$, and yield $\cO(k^{\cO(k)} n)$ and $\cO(k^{\cO(k^2)} n)$ time algorithms for the respective problems.

We give direct algorithmic consequences of the Merge Dominator Lemma in the realm of computing width parameters of directed acyclic graphs (DAGs).
Specifically, we show that the \textsc{(Weighted) Cutwidth} and \textsc{Modified Cutwidth} problems on DAGs,
which given a directed acyclic graph on $n$ vertices, ask for the topological order that minimizes the \emph{cutwidth} and \emph{modified cutwidth}, respectively,
can be solved in $\calO(n^2)$ time on \emph{series parallel digraphs}.
Note that the restriction of the solution to be a \emph{topological} order has been made as well in other works, e.g.~\cite{BFT09}.

Our algorithm for {\sc Cutwidth} of series parallel digraphs has
the same structure as the dynamic programming algorithm for undirected
{\sc Cutwidth}~\cite{BFT09}, but, in addition to obeying directions of edges, we have a step that only keeps characteristics that are not dominated by another characteristic in a table of characteristics. Now, with help of our Merge Dominator Lemma, we can show that in the case of series parallel digraphs, there is a unique dominating characteristic; the dynamic programming algorithm reverts to computing for each intermediate graph a single `optimal partial solution'.
This strategy also works in the presence of edge weights, which gives the algorithm for the corresponding \textsc{Weighted Cutwidth} problem on series parallel digraphs.
Note that the cutwidth of a directed acyclic graph is at least the maximum indegree or outdegree of a vertex; e.g., a series parallel digraph formed by the parallel composition of $n-2$ paths with three vertices has
$n$ vertices and cutwidth $n-2$.
To compute the \emph{modified} cutwidth of a series parallel digraph, we give a linear-time reduction to the 
\textsc{Weighted Cutwidth} problem on series parallel digraphs.

This paper is organized as follows. In Section~\ref{section:preliminaries}, 
we give a number of preliminary definitions, and review existing results,
including several results on typical sequences from \cite{BK96}.
In Section~\ref{sec:merge:dominate}, we state and prove the main technical result of this work, the Merge Dominator Lemma. 
Section~\ref{sec:width:measures} gives our algorithmic applications of
this lemma, and shows that the directed cutwidth and directed modified cutwidth of a series parallel
digraph can be computed in polynomial time. Some final remarks are made
in the conclusions Section~\ref{section:conclusions}.

\section{Preliminaries}
\label{section:preliminaries}
We use the following notation. For two integers $a, b \in \bN$ with $a \le b$, we let $[a..b] \defeq \{a, a+1, \ldots, b\}$ and for $a > 0$, we let $[a] \defeq [1..a]$. 
If $X$ is a set of size $n$, then a \emph{linear order} is a bijection $\linord\colon X \to [n]$. Given a subset $X' \subseteq X$ of size $n' \le n$, we define the \emph{restriction of $\linord$ to $X'$} as the bijection $\linord|_{X'} \colon X' \to [n']$ which is such that for all $x', y' \in X'$, $\linord|_{X'}(x') < \linord|_{X'}(y')$ if and only if $\linord(x') < \linord(y')$.
\paragraph{Sequences and Matrices.} We denote the elements of a sequence $s$ by $s(1), \ldots, s(n)$. We denote the \emph{length} of $s$ by $\length(s)$, i.e.\ $\length(s) \defeq n$. For two sequences $a = a(1), \ldots, a(m)$ and $b = b(1), \ldots, b(n)$, we denote their \emph{concatenation} by $a \concat b = a(1), \ldots, a(m), b(1), \ldots, b(n)$. For two sets of sequences $A$ and $B$, we let $A \allconcat B \defeq \{a \concat b \mid a \in A \wedge b \in B\}$.
For a sequence $s$ of length $n$ and a set $X \subseteq [n]$, we denote by $s[X]$ the \emph{subsequence of $s$ induced by $X$}, i.e.\ let $X = \{x_1, \ldots, x_{m}\}$ be such that for all $i \in [m-1]$, $x_i < x_{i+1}$; then, $s[X] \defeq s(x_1), \ldots, s(x_m)$.
For $x_1, x_2 \in [n]$ with $x_1 \le x_2$, we use the shorthand `$s[x_1..x_2]$' for `$s[[x_1..x_2]]$'.

Let $\Omega$ be a set. A \emph{matrix} $M \in \Omega^{m \times n}$ over $\Omega$ is said to have $m$ rows and $n$ columns. 
For sets $X \subseteq [m]$ and $Y \subseteq [n]$, we denote by $M[X, Y]$ the \emph{submatrix} of $M$ \emph{induced} by $X$ and $Y$, which consists of all the entries from $M$ whose indices are in $X \times Y$.
For $x_1, x_2 \in [m]$ with $x_1 \le x_2$ and $y_1, y_2 \in [n]$ with $y_1 \le y_2$,
we use the shorthand `$M[x_1..x_2, y_1..y_2]$' for `'$M[[x_1..x_2], [y_1..y_2]]$'. 
For a sequence $s(1), s(2), \ldots, s(\ell)$ of indices of a matrix $M$, we let 
	\begin{align}
		M[s] \defeq M[s(1)], M[s(2)], \ldots, M[s(\ell)]\label{eq:path:matrix}
	\end{align}
be the corresponding sequence of entries from $M$.

For illustrative purposes we enumerate the columns of a matrix in a bottom-up fashion throughout this paper, 
i.e.\ we consider the index $(1, 1)$ as the `bottom left corner' and the index $(m, n)$ as the `top right corner'.

\paragraph{Integer Sequences.}
Let $s$ be an integer sequence of length $n$.
We use the shorthand `$\min(s)$' for `$\min_{i \in [n]} s(i)$' and `$\max(s)$' for `$\max_{i \in [n]} s(i)$'; we use the following definitions.
We let 
\[
	\argmin(s) \defeq \{i \in [n] \mid s(i) = \min(s)\} \mbox{ and }
	\argmax(s) \defeq \{i \in [n] \mid s(i) = \max(s)\} 
\]
be the set of indices at whose positions there are the minimum and maximum element of $s$, respectively.
Whenever we write $i \in \argmin(s)$ ($j \in \argmax(s)$), then the choice of $i$ ($j$) can be arbitrary.
In some places we require a canonical choice of the position of a minimum or maximum element,
in which case we will always choose the smallest index. Formally, we let
\[
	\argmincan(s) \defeq \min \argmin(s), \mbox{ and } 
	\argmaxcan(s) \defeq \min \argmax(s).
\]

The following definition contains two notions on pairs of integer sequences 
that are necessary for the definitions of domination and merges.
\begin{definition}
	Let $r$ and $s$ be two integer sequences of the same length $n$.
	\begin{enumerate}
		\item If for all $i \in [n]$, $r(i) \le s(i)$, then we write `$r \le s$'.
		\item We write $q = r \merge s$ for the integer sequence $q(1), \ldots, q(n)$ with $q(i) = r(i) + s(i)$ for all $i \in [n]$.
	\end{enumerate}
\end{definition} 

\begin{definition}[Extensions]
	Let $s$ be a sequence of length $n$. We define the set $\extensions(s)$ of \emph{extensions} of $s$ as the set of sequences that are obtained from $s$ by repeating each of its elements an arbitrary number of times, and at least once. Formally, we let
	\[
		\extensions(s) \defeq \{s_1 \concat s_2 \concat \cdots \concat s_n \mid \forall i \in [n] \colon 
			\length(s_i) \ge 1 \wedge \forall j \in [\length(s_i)] \colon s_i(j) = s(i)\}.
	\]
\end{definition}

\begin{definition}[Domination]
	Let $r$ and $s$ be integer sequences. We say that \emph{$r$ dominates $s$}, in symbols `$r \dominates s$', if there are extensions $r^* \in \extensions(r)$ and $s^* \in \extensions(s)$ of the same length such that $r^* \le s^*$. If $r \dominates s$ and $s \dominates r$, then we say that $r$ and $s$ are \emph{equivalent}, and we write $r \equiv s$.
	
	If $r$ is an integer sequence and $S$ is a set of integer sequences, then we say that \emph{$r$ dominates $S$}, in symbols `$r \dominates S$', if for all $s \in S$, $r \dominates s$.
\end{definition}

\begin{remark}[Transitivity of `$\dominates$']
	In \cite[Lemma 3.7]{BK96}, it is shown that the relation `$\dominates$' is transitive. As this is fairly intuitive, we may use this fact without stating it explicitly throughout this text.
\end{remark}

\begin{definition}[Merges]
	Let $r$ and $s$ be two integer sequences. We define the set of all \emph{merges} of $r$ and $s$, denoted by $r \allmerges s$, as
	$
		r \allmerges s \defeq \{r^* + s^* \mid r^* \in \extensions(r), s^* \in \extensions(s), \length(r^*) = \length(s^*)\}.
	$
\end{definition}
\subsection{Typical Sequences} 
We now define typical sequences, show how to construct them in linear time, 
and restate several lemmas due to Bodlaender and Kloks~\cite{BK96} that will be used throughout this text.
\begin{definition}
	Let $s = s(1), \ldots, s(n)$ be an integer sequence of length $n$. The \emph{typical sequence of $s$}, denoted by $\typseq(s)$, is obtained from $s$ by an exhaustive application of the following two operations:
	
	\medskip\noindent
	\textit{Removal of Consecutive Repetitions.} If there is an index $i \in [n - 1]$ such that $s(i) = s(i+1)$, then we 
		change the sequence $s$ from $s(1), \ldots, s(i), s(i+1), \ldots, s(n)$ to $s(1), \ldots, s(i), s(i+2), \ldots, s(n)$.
	
	\medskip\noindent
	\textit{Typical Operation.} If there exist $i, j \in [n]$ such that $j - i \ge 2$ and for all $i \le k \le j$, $s(i) \le s(k) \le s(j)$, or for all $i \le k \le j$, $s(i) \ge s(k) \ge s(j)$, then we change the sequence $s$ from $s(1), \ldots, s(i), s(i+1), \ldots, s(j), \ldots, s(n)$ to $s(1), \ldots, s(i), s(j), \ldots, s(n)$, i.e.\ we remove all elements (strictly) between index $i$ and $j$.
\end{definition}

\begin{figure}
	\centering
	\includegraphics[height=.15\textheight]{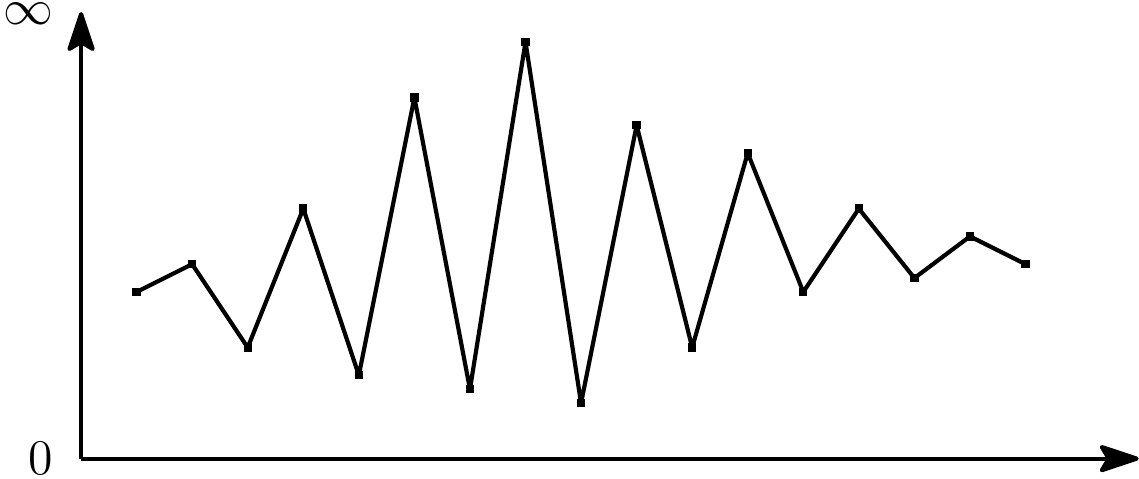}
	\caption{Illustration of the shape of a typical sequence.}
	\label{fig:typ:seq}
\end{figure}
To support intuition, we illustrate the rough shape of a typical sequence in \cref{fig:typ:seq}.
It is not difficult to see that the typical sequence can be computed in quadratic time,
by an exhaustive application of the definition.
Here we discuss how to do it in linear time.
We may view a typical sequence $\typseq(s)$ of an integer sequence $s$ as a subsequence of $s$. 
While $\typseq(s)$ is unique, the choice of indices that induce $\typseq(s)$ may not be unique.
We show that we can find a set of indices that induce the typical sequence in linear time,
with help of the following structural proposition.
\begin{proposition}\label{prop:comp:typseq:aux}
	Let $s$ be an integer sequence and let $i^\star \in \{\argmincan(s), \argmaxcan(s)\}$.
	Let $1 \eqdef j_0 < j_1 < j_2 < \ldots < j_t < j_{t+1} \defeq i^\star$ be pairwise distinct integers,
	such that $s(j_0), \ldots, s(j_{t+1})$ are pairwise distinct.
	If for all $h \in [0..t]$, 
	\begin{itemize}
		\item if $s(j_h) > s(j_{h+1})$ then $j_h = \argmaxcan(s[1..j_{h+1}])$ and $j_{h+1} = \argmincan(s[1..j_{h+1}])$, and
		\item if $s(j_h) < s(j_{h+1})$ then $j_h = \argmincan(s[1..j_{h+1}])$ and $j_{h+1} = \argmaxcan(s[1..j_{h+1}])$,
	\end{itemize}
	then the typical sequence of $s$ restricted to $[i^\star]$ is equal to $s(j_0), s(j_1), \ldots, s(j_t), s(j_{t+1})$.
\end{proposition}
\begin{proof}
	First, we observe that by the choice made in the definition of $\argmincan$ and $\argmaxcan$,
	\begin{align}
		\mbox{for each $h \in [0..(t+1)]$ there is no $i < j_h$ such that } s(i) = s(j_h).\label{prop:comp:typseq:aux:2}
	\end{align}
	
	We prove the following statement. Under the stated conditions, for a given $h \in [0..t+1]$, 
	the typical sequence of $s$ restricted to $[j_h..i^\star]$ is equal to $s(j_h)$, $s(j_{h+1})$, $\ldots$, $s(j_{t+1})$. 
	The proposition then follows from the case $h = 0$.
	The proof is by induction on $d \defeq (t+1) - h$. 
	For $d = 0$, it trivially holds since the minimum and the maximum element are always part of the typical sequence, 
	and since $[j_{t+1}..i^\star] = \{i^\star\}$.
	
	Now suppose $d > 0$, and for the induction hypothesis, that the claim holds for $d-1$.
	Suppose that $s(j_h) > s(j_{h+1})$, meaning that
	$j_h = \argmaxcan(s[1..j_{h+1}])$, and
	$j_{h+1} = \argmincan(s[1..j_{h+1}])$, the other case is symmetric.
	By the induction hypothesis, the typical sequence of $s$ 
	restricted to $[j_{h+1}..i^\star]$ is equal to $s(j_{h+1})$, $\ldots$, $s(j_{t+1})$,
	in particular it implies that $s(j_{h+1})$ is an element of the typical sequence.
	To prove the induction step, we have to show that the typical sequence 
	restricted to $[j_h..j_{h+1}]$ is equal to $s(j_h)$, $s(j_{h+1})$.
	We first argue that if there is an element of the typical sequence in $[j_h..(j_{h+1}-1)]$, then it must be equal to $s(j_h)$.	
	By \cref{prop:comp:typseq:aux:2}, we have that there is no $i < j_{h+1}$ such that $s(i) = s(j_{h+1})$, 
	hence $[j_h..(j_{h+1}-1)]$ cannot contain any element of the typical sequence that is equal to $s(j_{h+1})$.
	Next, since the typical operation removes all elements $i \in [(j_h+1)..(j_{h+1}-1)]$ with $s(j_h) > s(i) > s(j_{h+1})$, 
	and since $j_h = \argmaxcan(s[1..j_{h+1}])$,
	the only elements from $[j_h..(j_{h+1}-1)]$ that the typical sequence may contain have value $s(j_h)$.
	
	It remains to argue that $s(j_h)$ is indeed an element of the typical sequence. 
	Suppose not, then there are indices $i, i'$ with $i < j_h < i'$, such that either 
	$s(i) \le s(j_h) \le s(i')$, or
	$s(i) \ge s(j_h) \ge s(i')$, 
	and we may assume that at least one of the inequalities is strict in each case.
	For the latter case, since $j_h = \argmaxcan(s[1..j_{h+1}])$, 
	we would have that $s(i) = s(j_h)$, which is a contradiction to \cref{prop:comp:typseq:aux:2}.
	Hence, we may assume that $s(i) \le s(j_h) \le s(i')$.
	There are two cases to consider: $i' \in [(j_h+1)..j_{h+1}]$, and $i' > j_{h+1}$.
	If $i' \in [(j_h+1)..j_{h+1}]$, then $s(i') = s(j_h)$, as $s(j_h) = \argmax(s[1..j_{h+1}])$.
	We can conclude that in this case, the typical sequence must contain an element equal to 
	$s(i')$, and hence equal to $s(j_h)$.	
	If $i' > j_{h+1}$, then the typical operation corresponding to $i$ and $i'$ also removes $s(j_{h+1})$, 
	a contradiction with the induction hypothesis which asserts that 
	$s(j_{h+1})$ is part of the typical sequence induced by $[j_{h+1}..i^\star]$.
	We can conclude that $s(j_h)$ is part of the typical sequence, finishing the proof.
\end{proof}

From the previous proposition, we have the following consequence about the structure of 
typical sequences ending in the minimum element, which will be useful in the proof of \cref{lem:chop}.
\begin{corollary}\label{cor:typseq:min:shape}
	Let $t$ be a typical sequence of length $n$ such that $n \in \argmin(t)$. 
	Then, for each $k \in \left[\lfloor \frac{n}{2} \rfloor\right]$,
	$n - 2k + 1 \in \argmax(t[1..(n - 2k + 1)])$ and
	$n - 2k \in \argmin(t[1..(n - 2k)])$.
\end{corollary}

Equipped with \cref{prop:comp:typseq:aux}, we can now proceed and give the linear-time algorithm
that computes a typical sequence of an integer sequence.
\begin{lemma}\label{lem:compute:typseq}
	Let $s$ be an integer sequence of length $n$. Then, one can compute $\typseq(s)$, the typical sequence of $s$, in time $\cO(n)$.
\end{lemma}
\begin{proof}
	First, we check for each $i \in [n-1]$ whether $s(i) = s(i+1)$, and 
	if we find such an index $i$, we remove $s(i)$. 
	We assume from now on that after these modifications, 
	$s$ has at least two elements, otherwise it is trivial.
	As observed above, the typical sequence of $s$ contains $\min(s)$ and $\max(s)$.
	A closer look reveals the following observation.
	\begin{nestedobservation}\label{obs:comp:typseq:min:max}
		Let $i^\star \defeq \min \argmin(s) \cup \argmax(s)$ and $k^\star \defeq \max \argmin(s) \cup \argmax(s)$.
		\begin{enumerate}
			\item If $i^\star \in \argmin(s)$ and $k^\star \in \argmax(s)$ 
				or $i^\star \in \argmax(s)$ and $k^\star \in \argmin(s)$,
				then $\typseq(s)$ restricted to $[i^\star..k^\star]$ is equal to $s(i^\star), s(k^\star)$.
			\item If $\{i^\star, k^\star\} \subseteq \argmin(s)$, 
				then $\typseq(s)$ restricted to $[i^\star..k^\star]$ is equal to $s(i^\star), \max(s), s(k^\star)$.
			\item If $\{i^\star, k^\star\} \subseteq \argmax(s)$, 
				then $\typseq(s)$ restricted to $[i^\star..k^\star]$ is equal to $s(i^\star), \min(s), s(k^\star)$.
		\end{enumerate}
	\end{nestedobservation}
	
	Let $i^\star \defeq \min \argmin(s) \cup \argmax(s)$ and $k^\star \defeq \max \argmin(s) \cup \argmax(s)$.
	Using \cref{obs:comp:typseq:min:max}, it remains to determine the indices that induce the typical sequence on
	$s[1..i^\star]$ and on $s[k^\star..n]$.
	To find the indices that induce the typical sequence on $s[1..i^\star]$,
	we will describe a marking procedure that marks a set of indices satisfying 
	the preconditions of \cref{prop:comp:typseq:aux}.
	Next, we observe that $n - k^\star$ is the \emph{smallest} index of any occurrence of 
	$\min(s)$ or $\max(s)$ in the \emph{reverse} sequence of $s$, 
	therefore a symmetric procedure, again using \cref{prop:comp:typseq:aux}, 
	yields the indices that induce $\typseq(s)$ on $s[k^\star..n]$.

	\begin{algorithm}[h]
		\DontPrintSemicolon
		$j_{\min} \gets \argmincan(s[1..2])$, $j_{\max} \gets \argmaxcan(s[1..2])$, $M \gets \{1\}$\;
		\For{$j = 3, \ldots, i^\star$}{
		\If{$s(j) < s(j_{\min})$\label{alg:comp:typseq:cond:1}}{
				$j_{\min} \gets j$\;
				$M \gets M \cup \{j_{\max}\}$ \tcp{mark the current value of $j_{\max}$}
			}
			\If{$s(j) > s(j_{\max})$\label{alg:comp:typseq:cond:2}}{
				${j_{\max}} \gets j$\;
				$M \gets M \cup \{j_{\min}\}$ \tcp{mark the current value of $j_{\min}$}
			}
		}
		$M \gets M \cup \{j_{\min}, j_{\max}\}$\;
		\caption{The algorithm of \cref{lem:compute:typseq} that computes the set $M$ of indices that induce the 
		typical sequence of $s$ between the first element and the first occurrence of the minimum and maximum of $s$.}
		\label{alg:compute:typseq}
	\end{algorithm}	
	
	We execute \cref{alg:compute:typseq}, which processes 
	the integer sequence $s[1..i^\star]$ from
	the first to the last element,
	storing two counters $j_{\min}$ and $j_{\max}$
	that store the leftmost position of the smallest and of the greatest element
	seen so far, respectively.
	Whenever a new minimum is encountered, we mark the current value of the index $j_{\max}$, 
	as this implies that $s(j_{\max})$ has to be an element of the typical sequence.
	Similarly, when encountering a new maximum, we mark $j_{\min}$.
	These marked indices are stored in a set $M$, which at the end of the algorithm contains
	the indices that induce $\typseq(s)$ on $[1..i^\star]$.
	This, i.e.\ the correctness of the procedure, will now be argued via \cref{prop:comp:typseq:aux}.
	\begin{nestedclaim}\label{claim:comp:typseq:cor}
		The set $M$ of indices marked by the above procedure induce $\typseq(s)$ on $[1..i^\star]$. 
	\end{nestedclaim}
	\begin{claimproof}
		Let $M = \{j_0, j_1, \ldots, j_{t+1}\}$ be such that for all $h \in [0..t]$, $j_h < j_{h+1}$.
		We prove that $j_0, \ldots, j_{t+1}$ meet the preconditions of \cref{prop:comp:typseq:aux}.
		First, we observe that the above algorithm marks both the index $1$ and index $i^\star$, 
		in particular that $j_0 = 1$ and $j_{t+1} = i^\star$.		
			
		We verify that the indices $j_0, \ldots, j_{t+1}$ satisfy the property that for each $[0..(t+1)]$,
		the index $j_h$ is the leftmost (i.e.\ smallest) index whose value is equal to $s(j_h)$:
		whenever an index is added to the marked set, it is because in some iteration, 
		the element at its position was either strictly greater than the greatest previously seen element, 
		or strictly smaller than the smallest previously seen element.
		(This also ensures that $s(j_0), \ldots, s(j_{t+1})$ are pairwise distinct.)
		
		We additionally observe that if we have two indices $\ell_1$ and $\ell_2$ such that 
		$\ell_2$ is the index that the algorithm marked right after it marked $\ell_1$,
		then either $\ell_1$ was $j_{\min}$ and $\ell_2$ was $j_{\max}$ or vice versa: 
		when updating $j_{\min}$, we mark $j_{\max}$, and when updating $j_{\max}$, we mark $j_{\min}$.
		This lets us conclude that when we have two indices $j_h, j_{h+1}$ such that $s(j_h) < s(j_{h+1})$, then 
		$j_h$ was equal to $j_{\min}$ when it was marked, and $j_{h+1}$ was $j_{\max}$ when it was marked.
		
		We are ready to prove that $j_0, \ldots, j_{t+1}$ satisfy the precondition of \cref{prop:comp:typseq:aux}.
		Suppose for a contradiction that for some $h \in [0..t+1]$, $j_h$ violates this property.
		Assume that $s(j_h) < s(j_{h+1})$ and note that the other case is symmetric.
		The previous paragraph lets us conclude that $j_h$ was equal to $j_{\min}$ when it was marked,
		and that $j_{h+1}$ was $j_{\max}$ when it was marked.
		
		We may assume that either $j_h \neq \argmincan(s[1..j_{h+1}])$ or that $j_{h+1} \neq \argmaxcan(s[1..j_{h+1}])$.
		Suppose the latter holds. 
		This immediately implies that there is some $j^* \in [j_{h+1}-1]$ such that $s(j^*) > j_{h+1}$, 
		which implies that $j_{\max}$ would never have been set to $j_{h+1}$
 		and hence $j_{h+1}$ would have never been marked.
		Suppose the former holds, i.e.\ $j_h \neq \argmincan(s[1..j_h])$, 
		for an illustration of the following argument see \cref{fig:compute:typseq:cor}.
		\begin{figure}
			\centering
			\includegraphics[height=.225\textheight]{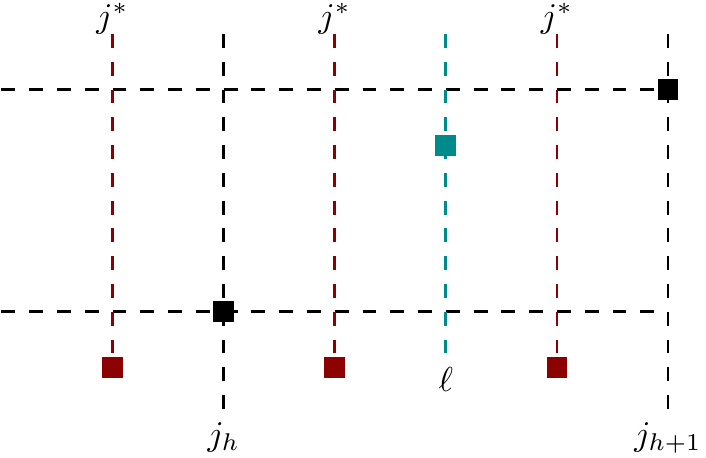}
			\caption{Illustration of the final argument in the proof of \cref{claim:comp:typseq:cor}. 
				We assume that $s(j_h) < s(j_{h+1})$, 
				and mark the possible positions for $j^* = \argmincan(s[1..j_{h+1}])$ with $j^* \neq j_h$.}
			\label{fig:compute:typseq:cor}
		\end{figure}	
		Let $j^* \defeq \argmincan(s[1..j_{h+1}])$.
		If $j^* < j_h$, then at iteration $j_h$, $s(j_{\min}) < s(j_h)$, so $j_{\min}$ would never have been set to $j_h$,
		and hence, $j_h$ would never have been marked. We may assume that $j^* > j_h$.
		Since $j_h$ was marked, there is some $\ell > j_h$ that triggered $j_h$ being marked. 
		This also means that at that iteration $s(\ell)$ was greater than the previously observed maximum,
		so we may assume that $s(\ell) > s(j_h)$.		
		We also may assume that $\ell \le j_{h+1}$. 
		If $j^* \in [(j_h+1)..(\ell-1)]$, then the algorithm would have updated $j_{\min}$ to $j^*$ in that iteration, 
		before marking $j_h$,
		and for the case $j^* \in [(\ell+1)..(j_{h+1}-1)]$ we observe that $\ell \neq j_{h+1}$, and that
		the algorithm would mark $\ell$ as the next index instead of $j_{h+1}$.
	\end{claimproof}
	
	This establishes the correctness of the algorithm. For its runtime, we observe that each iteration takes $\cO(1)$ time, and that there are $\cO(n)$ iterations.
\end{proof}

We summarize several lemmas from~\cite{BK96} regarding 
integer sequences and typical sequences that we will use in this work.
\begin{lemma}[Bodlaender and Kloks \cite{BK96}]\label{lem:BK}
	Let $r$ and $s$ be two integer sequences.
	\begin{enumerate}[label={(\roman*)}]
		\item\enumheader{(Cor.\ 3.11 in~\cite{BK96}).}\label{cor:BK:3:11} We have that $r \dominates s$ if and only if $\typseq(r) \dominates \typseq(s)$.
		\item\enumheader{(Lem.\ 3.13 in~\cite{BK96}).}\label{lem:BK:3:13} Suppose $r$ and $s$ are of the same length and let $y = r \merge s$. Let $r_0 \dominates r$ and $s_0 \dominates s$. Then there is an integer sequence $y_0 \in r_0 \allmerges s_0$ such that $y_0 \dominates y$.
		\item\enumheader{(Lem.\ 3.14 in~\cite{BK96}).}\label{lem:BK:3:14} Let $q \in r \allmerges s$. Then, there is an integer sequence $q' \in \typseq(r) \allmerges \typseq(s)$ such that $q' \dominates q$.
		\item\enumheader{(Lem.\ 3.15 in~\cite{BK96}).}\label{lem:BK:3:15} Let $q \in r \allmerges s$. Then, there is an integer sequence $q' \in r \allmerges s$ with $\typseq(q') = \typseq(q)$ and $\length(q') \le \length(r) + \length(s) -1$.
		\item\enumheader{(Lem.\ 3.19 in~\cite{BK96}).}\label{lem:BK:3:19} Let $r'$ and $s'$ be two more integer sequences. If $r' \dominates r$ and $s' \dominates s$, then $r' \concat s' \dominates r \concat s$.
	\end{enumerate}
\end{lemma}

\subsection{Directed Acyclic Graphs} 
A \emph{directed graph} (or \emph{digraph}) $G$ is a pair of a set of \emph{vertices} $V(G)$ 
and a set of ordered pairs of vertices, called \emph{arcs}, $A(G) \subseteq V(G) \times V(G)$. 
(If $A(G)$ is a multiset, we call $G$ \emph{multidigraph}.) 
We say that an arc $a = (u, v) \in A(G)$ is directed from $u$ to $v$, and we call $u$ the \emph{tail} of $a$ and $v$ the \emph{head} of $a$. We use the shorthand `$uv$' for `$(u, v)$'. A sequence of vertices $v_1, \ldots, v_r$ is called a \emph{walk} in $G$ if for all $i \in [r-1]$, $v_i v_{i+1} \in A(G)$. A \emph{cycle} is a walk $v_1, \ldots, v_r$ with $v_1 = v_r$ and all vertices $v_1, \ldots, v_{r-1}$ pairwise distinct.
If $G$ does not contain any cycles, then we call $G$ \emph{acyclic} or a \emph{directed acyclic graph}, DAG for short.

Let $G$ be a DAG on $n$ vertices. A \emph{topological order} of $G$ is a linear order $\linord \colon V(G) \to [n]$ such that for all arcs $uv \in A(G)$, we have that $\linord(u) < \linord(v)$. We denote the set of all topological orders of $G$ by $\toporders(G)$.
We now define the width measures studied in this work. Note that we restrict the orderings of the vertices that we consider to \emph{topological} orderings.
\begin{definition}\label{def:dag:measures}
	Let $G$ be a directed acyclic graph and let $\linord \in \toporders(G)$ be a topological order of $G$.
	\begin{enumerate}
		\item\label{def:dag:measures:cutwidth} The \emph{cutwidth} of $\linord$ is $\cutwidth(\linord) \defeq \max_{i \in [n-1]} \card{\{uv \in A(G) \mid \linord(u) \le i \wedge \linord(v) > i\}}$.
		\item\label{def:dag:measures:mod:cutwidth} The \emph{modified cutwidth} of $\linord$ is $\modifiedcutwidth(\linord) \defeq \max_{i \in [n]} \card{\{uv \in A(G) \mid \linord(u) < i \wedge \linord(v) > i\}}$.
	\end{enumerate}
	We define the cutwidth and modified cutwidth of a directed acyclic graph $G$ as the minimum of the respective measure over all topological orders of $G$.
\end{definition}

We now introduce series parallel digraphs. Note that the following definition coincides with the notion of `edge series-parallel multidigraphs' in~\cite{VTL82}.
For an illustration see \cref{fig:spd}.
\begin{figure}
	\centering
	\includegraphics[width=.75\textwidth]{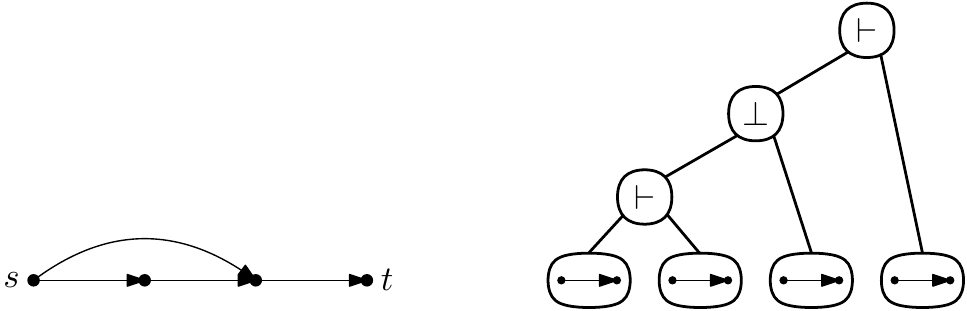}
	\caption{A series parallel digraph $G$ on the left, and a decomposition tree that yields $G$ on the right.}
	\label{fig:spd}
\end{figure}
\begin{definition}[Series Parallel Digraph (SPD)]
	A (multi-)digraph $G$ with an ordered pair of \emph{terminals} $(s, t) \in V(G) \times V(G)$ is called \emph{series parallel digraph (SPD)}, often denoted by $(G, (s, t))$, if one of the following hold.
	\begin{enumerate}
		\item $(G, (s, t))$ is a single arc directed from $s$ to $t$, i.e.\ $V(G) = \{s, t\}$, $A(G) = \{(s, t)\}$.
		\item $(G, (s, t))$ can be obtained from two series parallel digraphs $(G_1, (s_1, t_1))$ and $(G_2, (s_2, t_2))$ by one of the following operations.
		\begin{enumerate}
			\item\enumheader{Series Composition.} $(G, (s, t))$ is obtained by taking the disjoint union of $G_1$ and $G_2$, identifying $t_1$ and $s_2$, and letting $s = s_1$ and $t = t_2$. In this case we write $(G, (s, t)) = (G_1, (s_1, t_1)) \seriescomp (G_2, (s_2, t_2))$ or simply $G = G_1 \seriescomp G_2$.
			\item\enumheader{Parallel Composition.} $(G, (s, t))$ is obtained by taking the disjoint union of $G_1$ and $G_2$, identifying $s_1$ and $s_2$, and identifying $t_1$ and $t_2$, and letting $s = s_1 = s_2$ and $t = t_1 = t_2$. In this case we write $(G, (s, t)) = (G_1, (s_1, t_1)) \parallelcomp (G_2, (s_2, t_2))$, or simply $G = G_1 \parallelcomp G_2$.
		\end{enumerate}
	\end{enumerate}
\end{definition}

It is not difficult to see that each series parallel digraph is acyclic.
One can naturally associate a notion of \emph{decomposition trees} with series parallel digraphs as follows. A decomposition tree $T$ is a rooted and ordered binary tree whose leaves are labeled with a single arc, and each internal node $t \in V(T)$ with left child $\ell$ and right child $r$ is either a \emph{series node} or a \emph{parallel node}. We then associate an SPD $G_t$ with $t$ that is $G_\ell \seriescomp G_r$ if $t$ is a series node and $G_\ell \parallelcomp G_r$ if $t$ is a parallel node. 
It is clear that for each SPD $G$, there is a decomposition tree $T$ with root $\fr$ such that $G = G_\fr$.
In that case we say that \emph{$T$ yields $G$}.
Valdes et al.~\cite{VTL82} have shown that one can decide in linear time whether a directed graph $G$ is an SPD and if so, find a decomposition tree that yields $G$.
\begin{theorem}[Valdes et al.~\cite{VTL82}]\label{thm:spd:recognition}
Let $G$ be a directed graph on $n$ vertices and $m$ arcs. There is an algorithm that decides in time $\cO(n + m)$ whether $G$ is a series parallel digraph and if so, it outputs a decomposition tree that yields $G$.
\end{theorem}

\section{The Merge Dominator Lemma}\label{sec:merge:dominate}
In this section we prove the main technical result of this work. It states that given two integer sequences, one can find in linear time a merge that dominates all merges of those two sequences.
\begin{lemma}[Merge Dominator Lemma]\label{lem:merge:dom}
	Let $r$ and $c$ be integer sequence of length $m$ and $n$, respectively. There exists a dominating merge of $r$ and $c$, i.e.\ an integer sequence $t \in r \allmerges c$ such that $t \dominates r \allmerges c$, and this dominating merge can be computed in time $\cO(m + n)$.
\end{lemma}

\paragraph{Outline of the proof of the Merge Dominator Lemma.}
First, we show that we can restrict our search to finding a dominating path in a matrix that, roughly speaking, contains all merges of $r$ and $c$ of length at most $\length(r) + \length(c) - 1$. The goal of this step is mainly to increase the intuitive insight to the proofs in this section. Next, we prove the `Split Lemma' (\cref{lem:split} in \cref{sec:lem:split}) which asserts that we can obtain a dominating path in our matrix $M$ by splitting $M$ into a submatrix $M_1$ that lies in the `bottom left' of $M$ and another submatrix $M_2$ in the `top right' of $M$ along a minimum row and a minimum column, and appending a dominating path in $M_2$ to a dominating path in $M_1$.
In $M_1$, the last row and column are a minimum row and column, respectively, and in $M_2$, the first row and column are a minimum row and column, respectively.
This additional structure will be exploited in \cref{sec:lem:chop} where we prove the `Chop Lemmas' that come in two versions.
The `bottom version' (\cref{lem:chop}) shows that in $M_1$, we can find a dominating path by repeatedly chopping away the \emph{last} two rows or columns and remembering a vertical or horizontal length-$2$ path.
The `top version' (\cref{lem:chop:top}) is the symmetric counterpart for $M_2$.
The proofs of the Chop Lemmas only hold when $r$ and $c$ are \emph{typical sequences}, and in \cref{sec:alg:split-and-chop} we present the `Split-and-Chop Algorithm' that computes a dominating path in a merge matrix of two typical sequences. Finally, in \cref{sec:split-chop:arbitrary}, we generalize this result to arbitrary integer sequences, using the Split-and-Chop Algorithm and one additional construction.

\subsection{The Merge Matrix, Paths, and Non-Diagonality}
Let us begin by defining the basic notions of a merge matrix and paths in matrices.
\begin{definition}[Merge Matrix]
	Let $r$ and $c$ be two integer sequences of length $m$ and $n$, respectively. Then, the \emph{merge matrix} of $r$ and $c$ is an $m \times n$ integer matrix $M$ such that for $(i, j) \in [m] \times [n]$, $M[i, j] = r(i) + c(j)$.
\end{definition}
\begin{definition}[Path in a Matrix]
	Let $M$ be an $m \times n$ matrix. A \emph{path} in $M$ is a sequence $p(1), \ldots, p(\ell)$ of indices from $M$ such that 
	\begin{enumerate}
		\item $p(1) = (1, 1)$ and $p(\ell) = (m, n)$, and
		\item for $h \in [\ell-1]$, let $p(h) = (i, j)$; then, $p(h+1) \in \{(i + 1, j), (i, j+1), (i + 1, j + 1)\}$.
	\end{enumerate}
	We denote by $\allpaths(M)$ the set of all paths in $M$. A sequence $p(1), \ldots, p(\ell)$ that satisfies the second condition but not necessarily the first is called a \emph{partial path} in $M$.
	For two paths $p, q \in \allpaths(M)$, we may simply say that
	\emph{$p$ dominates $q$}, if $M[p]$ dominates $M[q]$.\footnote{Recall that 
	by \eqref{eq:path:matrix} on page~\pageref{eq:path:matrix},
	for a (partial) path $p$ in a matrix $M$, $M[p] = M[p(1)], M[p(2)], \ldots, M[p(\length(p))]$.} 
	We also write $p \dominates \allpaths(M)$ to express that for each path $q \in \allpaths(M)$,
	$p \dominates q$.
	
	A (partial) path is called \emph{non-diagonal} if the second condition is replaced by the following.%
	\begin{enumerate}[label={(\roman*)'}]
		\setcounter{enumi}{1}
		\item For $h \in [\ell-1]$, let $p(h) = (i, j)$; then, $p(h+1) \in \{(i + 1, j), (i, j+1)\}$.
	\end{enumerate}
\end{definition}

An \emph{extension} $e$ of a path $p$ in a matrix $M$ is as well a sequence of indices of $M$, 
and we again denote the corresponding integer sequence by $M[e]$.
A consequence of \cref{lem:BK}\cref{cor:BK:3:11,lem:BK:3:15} is that we can restrict ourselves to all paths in a merge matrix when trying to find a dominating merge of two integer sequences: it is clear from the definitions that in a merge matrix $M$ of integer sequences $r$ and $c$, $\allpaths(M)$ contains all merges of $r$ and $c$ of length at most $\length(r) + \length(c) - 1$.
Furthermore, suppose that there is a merge $q \in r \allmerges s$ such that $q \dominates r \allmerges s$ and $\length(q) > \length(r) + \length(s) - 1$.
By \cref{lem:BK}\cref{lem:BK:3:15}, there is a merge $q' \in r \allmerges s$ such that $\length(q') \le \length(r) + \length(s) - 1$,
and $\typseq(q') = \typseq(q)$. The latter yields $\typseq(q') \equiv \typseq(q)$ 
and therefore, by \cref{lem:BK}\cref{cor:BK:3:11}, $q' \equiv q$, in particular, $q' \dominates q \dominates r \allmerges s$.
\begin{corollary}\label{cor:merge:path}
	Let $r$ and $c$ be integer sequences and $M$ be the merge matrix of $r$ and $c$. There is a dominating merge in $r \allmerges c$, i.e.\ an integer sequence $t \in r \allmerges c$ such that $t \dominates r \allmerges c$, if and only if there is a dominating path in $M$, i.e.\ a path $p \in \allpaths(M)$ such that $p \dominates \allpaths(M)$.
\end{corollary}

We now consider a type of merge that corresponds to non-diagonal paths in the merge matrix. 
These merges will be used in a construction presented in \cref{sec:split-chop:arbitrary}, 
and in the algorithmic applications of the Merge Dominator Lemma given in \cref{sec:width:measures}.
For two integer sequences $r$ and $s$, we denote by $r \allmergesnd s$ the set of all \emph{non-diagonal merges} of $r$ and $s$, which are not allowed to have `diagonal' steps: we have that for all $t \in r \allmergesnd s$ and all $i \in [\length(t)-1]$, if $t(i) = r(i_r) + s(i_s)$, then $t(i+1) \in \{r(i_r + 1) + s(i_s), r(i_r) + s(i_s + 1)\}$. 
As each non-diagonal merge directly corresponds to a non-diagonal path in the merge matrix (and vice versa),
we can consider a non-diagonal path in a merge matrix to be a non-diagonal merge and vice versa.
We now show that for each merge that uses diagonal steps, there is always a non-diagonal merge that dominates it.
\begin{lemma}\label{lem:merge:nd:equiv}
	Let $r$ and $s$ be two integer sequences of length $m$ and $n$, respectively. For any merge $q \in r \allmerges s$, there is a non-diagonal merge $q' \in r \allmergesnd s$ such that $q' \dominates q$. Furthermore, given $q$, $q'$ can be found in time $\calO(m + n)$.
\end{lemma}
\begin{proof}
	This can be shown by the following local observation. Let $i \in [\length(q) - 1]$ be such that $q(i), q(i+1)$ is a diagonal step, i.e.\ there are indices $i_r \in [\length(r) - 1]$ and $i_s \in [\length(s)-1]$ such that $q(i) = r(i_r) + s(i_s)$ and $q(i+1) = r(i_r + 1) + s(i_s + 1)$. Then, we insert the element $x \defeq \min\{r(i_r) + s(i_s + 1), r(i_r + 1) + s(i_s)\}$ between $q(i)$ and $q(i+1)$. 
	Since 
	$$x \le \max\{r(i_r) + s(i_s), r(i_r + 1), s(i_s + 1)\} \eqdef y,$$
	we can repeat $y$ twice in an extension of $q$ so that one of the occurrences aligns with $x$,
	and we have that in this position, the value of $q'$ is at most the value of the extension of $q$.
	
	Let $q'$ be the sequence obtained from $q$ by applying this operation to all diagonal steps,
	then by the observation just made, we have that $q' \dominates q$. 
	It is clear that this can be implemented to run in time $\calO(m + n)$.
\end{proof}

Next, we define two special paths in a matrix $M$ that will reappear in several places throughout this section. These paths can be viewed as the `corner paths', where the first one follows the first row until it hits the last column and then follows the last column ($\rightup{p}(M)$), and the second one follows the first column until it hits the last row and then follows the last row ($\upright{p}(M)$). Formally, we define them as follows:
	\begin{align*}
		\rightup{p}(M) &\defeq (1, 1), (1, 2), \ldots, (1, n), (2, n), \ldots, (m, n) \\
		\upright{p}(M) &\defeq (1, 1), (2, 1) \ldots, (m, 1), (m, 2), \ldots, (m, n)
	\end{align*}
	We use the shorthands `$\rightup{p}$' for `$\rightup{p}(M)$' and `$\upright{p}$' for `$\upright{p}(M)$' whenever $M$ is clear from the context.
\begin{figure}
	\centering
	\includegraphics[height=.3\textheight]{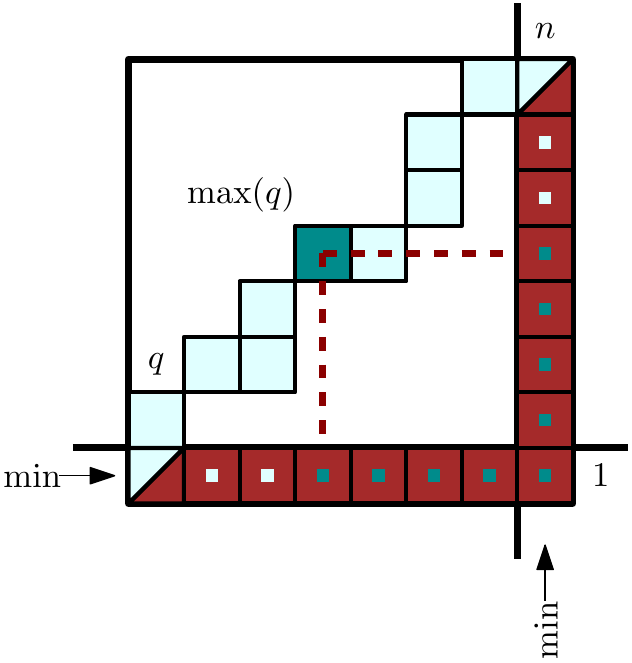}
	\caption{Situation in the proof of \cref{lem:special}\cref{lem:rightup}. 
		The dot within each element of the corner path $\rightup{p}$ indicates with which elements 
		of the path $q$ it is `matched up' in the extensions constructed in the proof.}
	\label{fig:rightup}
\end{figure}

For instance, these paths appear in the following special cases of the Merge Dominator Lemma, which will be useful for several proofs in this section.
\begin{lemma}\label{lem:special}
	Let $r$ and $c$ be integer sequences of length $m$ and $n$, respectively, and let $M$ be the merge matrix of $r$ and $c$. Let $i \in \argmin(r)$ and $j \in \argmin(c)$.
	\begin{enumerate}
		\item\label{lem:rightup} If $i = 1$ and $j = n$, then $\rightup{p}$ dominates all paths in $M$, i.e.\ $\rightup{p} \dominates \allpaths(M)$.
		\item\label{lem:upright} If $i = m$ and $j = 1$, then $\upright{p}$ dominates all paths in $M$, i.e.\ $\upright{p} \dominates \allpaths(M)$.
	\end{enumerate}
\end{lemma}
\begin{proof}
	\cref{lem:rightup} For an illustration of this proof see \cref{fig:rightup}. 
	Let $q$ be any path in $M$ and let $t^* \defeq \argmaxcan(q)$. Let furthermore $q(t^*) = (t^*_r, t^*_c)$.
	We divide $\rightup{p}$ and $q$ in three consecutive parts each to show that $\rightup{p}$ dominates $q$.
\begin{itemize}
	\item We let $\rightup{p}^1 \defeq \rightup{p}(1), \ldots, \rightup{p}(t^*_c - 1)$ and $q_1 \defeq q(1), \ldots, q(t^* - 1)$.
	\item We let $\rightup{p}^2 \defeq \rightup{p}(t^*_c), \ldots, \rightup{p}(n + t^*_r - 1)$ and $q_2 \defeq q(t^*)$.
	\item We let $\rightup{p}^3 \defeq \rightup{p}(n + t^*_r), \ldots, \rightup{p}(m + n - 1)$ and $q_3 \defeq q(t^* + 1), \ldots, q(\length(q))$.
\end{itemize}

	Since $r(1)$ is a minimum row in $M$, we have that for all $(k, \ell) \in [m] \times [n]$, $M[1, \ell] \le M[k, \ell]$. This implies that there is an extension $e_1$ of $\rightup{p}^1$ of length $t^* - 1$ such that $M[e_1] \le M[q_1]$. 
	Similarly, there is an extension $e_3$ of $\rightup{p}^3$ of length $\length(q) - t^*$ such that $M[e_3] \le M[q_3]$.
	Finally, let $f_2$ be an extension of $q_2$ that repeats its only element, $q(t^*)$, $n - t_c^* + t_r^*$ times. 
	Since $M[q(t^*)]$ is the maximum element on the sequence $M[q]$ and $r(1)$ is a minimum row and $c(n)$ a minimum column in $M$, 
	we have that $M[\rightup{p}^2] \le M[f_2]$.
	
	We define an extension $e$ of $\rightup{p}$ as $e \defeq e_1 \concat \rightup{p}^2 \concat e_3$ and an extension $f$ of $q$ as $f \defeq q_1 \concat f_2 \concat q_3$. Note that $\length(e) = \length(f) = \length(q) + n + t_r^* - (t_c^* + 1)$, and by the above discussion, we have that $M[e] \le M[f]$.
	\cref{lem:upright} follows from a symmetric argument.
\end{proof}

\subsection{The Split Lemma}\label{sec:lem:split}
In this section we prove the first main step towards the Merge Dominator Lemma. It is fairly intuitive that a dominating merge has to contain the minimum element of a merge matrix. (Otherwise, there is a path that cannot be dominated by that merge.) The Split Lemma states that in fact, we can split the matrix $M$ into two smaller submatrices, one that has the minimum element in the top right corner, and one the has the minimum element in the bottom left corner, compute a dominating path for each of them, and paste them together to obtain a dominating path for $M$.

\begin{figure}
	\centering
	\includegraphics[height=.3\textheight]{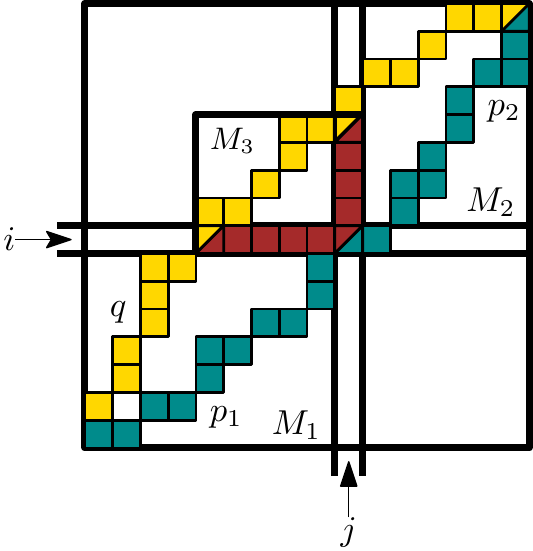}
	\caption{Situation in the proof of \cref{lem:split}.}
	\label{fig:split:bottom}
\end{figure}
\begin{lemma}[Split Lemma]\label{lem:split}
	Let $r$ and $c$ be integer sequences of length $m$ and $n$, respectively, and let $M$ be the merge matrix of $r$ and $c$. Let $i \in \argmin(r)$ and $j \in \argmin(c)$. Let $M_1 \defeq M[1..i,1..j]$ and $M_2 \defeq M[i..m,j..n]$ and for all $h \in [2]$, let $p_h \in \allpaths(M_t)$ be a dominating path in $M_h$, i.e.\ $p_h \dominates \allpaths(M_h)$. Then, $p_1 \concat p_2$ is a dominating path in $M$, i.e.\ $p_1 \concat p_2 \dominates \allpaths(M)$.
\end{lemma}
\begin{proof}
	Let $q$ be any path in $M$. If $q$ contains $(i, j)$, then $q$ has two consecutive parts, say $q_1$ and $q_2$, such that 
	$q_1 \in \allpaths(M_1)$ and $q_2 \in \allpaths(M_2)$. 
	Hence, $p_1 \dominates q_1$ and $p_2 \dominates q_2$, 
	so by \cref{lem:BK}\cref{lem:BK:3:19}, $p_1 \concat p_2 \dominates q_1 \concat q_2$.
	
	Now let $p \defeq p_1 \concat p_2$ and suppose $q$ does not contain $(i, j)$. Then, $q$ either contains some $(i, j')$ with $j' < j$, 
	or some $(i', j)$, for some $i' < i$. 
	We show how to construct extensions of $p$ and $q$ that witness that $p$ dominates $q$ in the first case, 
	and remark that the second case can be shown symmetrically. 
	We illustrate this situation in~\cref{fig:split:bottom}.
	
	Suppose that $q$ contains $(i, j')$ with $j' < j$. We show that $p \dominates q$.
	First,
	$q$ also contains some $(i', j)$, where $i' > i$. 
	Let $h_1$ be the index of $(i, j')$ in $q$, i.e.\ $q(h_1) = (i, j')$, 
	and $h_2$ denote the index of $(i', j)$ in $q$, i.e.\ $q(h_2) = (i', j)$.		
	We derive the following sequences from $q$.
	\begin{itemize}
		\item We let $q_1 \defeq q(1), \ldots, q(h_1)$ and $q_1^+ \defeq q_1 \concat (i, j' + 1), \ldots, (i, j)$.
		\item We let $q_{12} \defeq q(h_1), \ldots, q(h_2)$.
		\item We let $q_2 \defeq q(h_2), \ldots, q(\length(q))$ and $q_2^+ \defeq (i, j), (i + 1, j), \ldots, (i', j) \concat q_2$.
	\end{itemize}

	Since $q_1^+ \in \allpaths(M_1)$ and $p_1 \dominates \allpaths(M_1)$, we have that $p_1 \dominates q_1^+$, similarly that $p_2 \dominates q_2^+$ and considering $M_3 \defeq M[i'..i,j..j']$, we have by \cref{lem:special}\cref{lem:rightup} that 
	$p_{12} \defeq \rightup{p}(M_3) = (i, j'), (i, j' + 1), \ldots, (i, j), (i+1, j), \ldots, (i', j)$ dominates $q_{12}$. 
	Consequently, we consider the following extensions of these sequences.
	
	\begin{enumerate}[label={(\Roman*)},ref={(\Roman*)}]
		\item\label{claim:extension:1} We let $e_1 \in \extensions(p_1)$ and $f_1 \in \extensions(q_1^+)$ such that $\length(e_1) = \length(f_1)$ and $M[e_1] \le M[f_1]$.
		\item\label{claim:extension:12} We let $e_{12} \in \extensions(p_{12})$, and $f_{12} \in \extensions(q_{12})$ such that $\length(e_{12}) = \length(f_{12})$ and $M[e_{12}] \le M[f_{12}]$.
		\item\label{claim:extension:2} We let $e_2 \in \extensions(p_2)$, and $f_2 \in \extensions(q_2^+)$ such that $\length(e_2) = \length(f_2)$ and $M[e_2] \le M[f_2]$.
	\end{enumerate}

	We construct extensions $e' \in \extensions(p)$ and $f' \in \extensions(q)$ as follows. 
	Let $z$ be the last index in $q$ of any element that is matched up with $(i, j)$ 
	in the extensions of \cref{claim:extension:12}. 
	(Following the proof of \cref{lem:special}, this would mean $z$ is the index of $\max(q_{12})$ in $q$.)
	We first construct a pair of extensions $e_j' \in E(p_1)$, and $f_j' \in E(q[1..z])$ 
	with $\length(e_j') = \length(f_j')$ and $M[e_j'] \le M[f_j']$.
	With a symmetric procedure,
	we can obtain extensions of $p_2$ and of $q[(z+1)..\length(q)]$, and use them to obtain extensions of 
	$p = p_1 \concat p_2$ and $q = q[1..z] \concat q[(z+1)..\length(q)]$ witnessing that $p \dominates q$.
	
	We give the details of the first part of the construction.
	Let $a$ be the index of 
	the last repetition in $f_1$ of $q(h_1 - 1)$, i.e.\ the index that appears just before $q(h_1) = (i, j')$ in $f_1$. 
	We let $e_{j'-1}'[1..a] \defeq e_1[1..a]$ and $f_{j'-1}'[1..a] \defeq f_1[1..a]$. By \cref{claim:extension:1}, 
	$M[e_{j'-1}'] \le M[f_{j'-1}']$.
	\begin{figure}
		\centering
		\includegraphics[width=.85\textwidth]{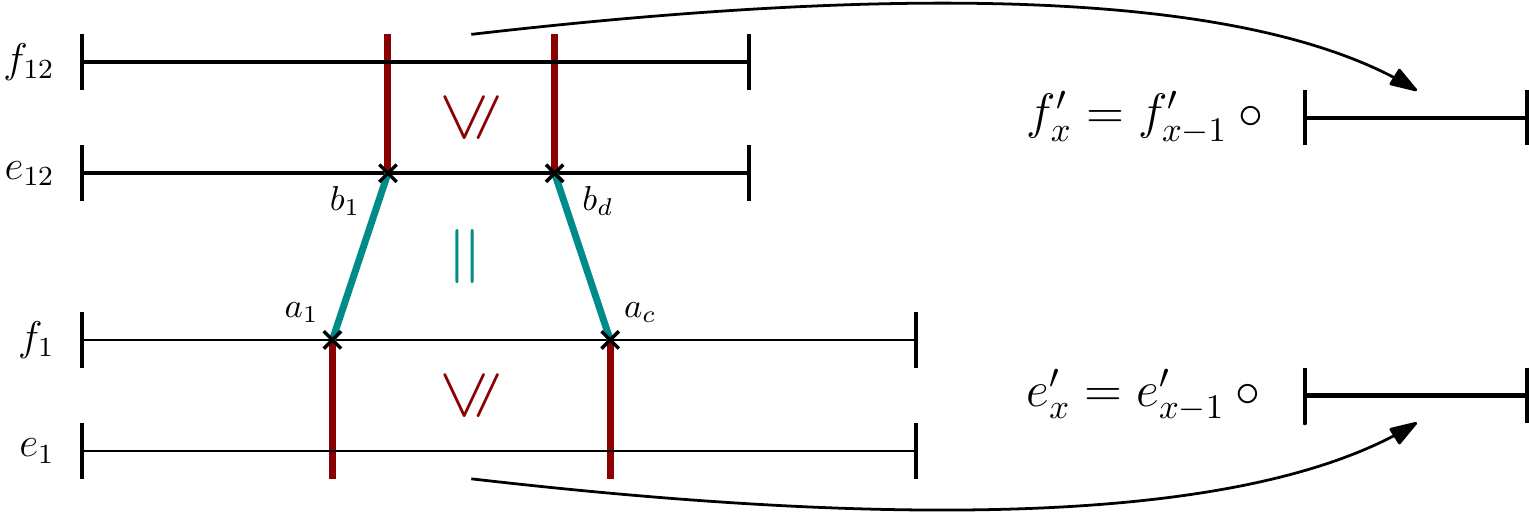}
		\caption{Constructing extensions in the proof of \cref{lem:split}.}
		\label{fig:stitching}
	\end{figure}
	
	For $x = j', j' + 1, \ldots, j$, we inductively construct $e_x'$ and $f_x'$ using $e_{x-1}'$ and $f_{x-1}'$, 
	for an illustration see~\cref{fig:stitching}. 
	We maintain as an invariant that $\length(e_{x-1}') = \length(f_{x-1}')$ and that $M[e_{x-1}'] \le M[f_{x-1}']$.
	Let $a_1, \ldots, a_c$ denote the indices of the occurrences of $(i, x)$ in $f_1$, 
	and $b_1, \ldots, b_d$ denote the indices of the occurrences of $(i, x)$ in $e_{12}$. We let:
	\begin{align*}
		\begin{array}{ll}
			e_x' \defeq e_{x-1}' \concat e_1[a_1, \ldots, a_c] \mbox{ and } f_x' \defeq f_{x-1}' \concat f_{12}[b_1, \ldots, b_d], 
				&\mbox{ if } c = d \\
			e_x' \defeq e_{x-1}' \concat e_1[a_1, \ldots, a_c] \concat \overbrace{e_1(a_c), \ldots, e_1(a_c)}^{d - c \mbox{ times}} \mbox{ and } f_x' \defeq f_{x-1}' \concat f_{12}[b_1, \ldots, b_d], 
			 	&\mbox{ if } c < d \\
			e_x' \defeq e_{x-1}' \concat e_1[a_1, \ldots, a_c] \mbox{ and } f_x' \defeq f_{x-1}' \concat f_{12}[b_1, \ldots, b_d] \concat \overbrace{f_{12}(b_d), \ldots, f_{12}(b_d)}^{c - d \mbox{ times}}, &\mbox{ if } c > d
		\end{array}
	\end{align*}

	In each case, we extended $e_{x-1}'$ and $f_{x-1}'$ by the same number of elements; furthermore we know by \cref{claim:extension:1} that for $y \in \{a_1, \ldots, a_c\}$, $M[e_1(y)] \le M[f_1(y)]$, by choice we have that for all $y' \in \{b_1, \ldots, b_d\}$, $f_1(y) = e_{12}(y')$ and we know that $M[e_{12}(y')] \le M[f_{12}(y')]$ by \cref{claim:extension:12}. 
	Hence, $M[e_x'] \le M[f_x']$ in either of the above cases.
	In the end of this process, we have $e_j' \in \extensions(p_1)$ and 
	$f_j' \in \extensions(q[1..z])$, and
	by construction, $\length(e_j') = \length(f_j')$ and $M[e_j'] \le M[f_j']$.
\end{proof}

\subsection{The Chop Lemmas}\label{sec:lem:chop}

Assume the notation of the Split Lemma. If we were to apply it recursively, it only yields a size reduction whenever $(i, j) \notin \{(1, 1), (m, n)\}$. Motivated by this issue, we prove two more lemmas to deal with the cases when $(i, j) \in \{(1, 1), (m, n)\}$, and we coin them the `Chop Lemmas'. It will turn out that when applied to typical sequences, a repeated application of these lemmas yields a dominating path in $M$. This insight crucially helps in arguing that the dominating path in a merge matrix can be found in \emph{linear} time.
Before we present their statements and proofs, we need another auxiliary lemma.

\begin{lemma}\label{lem:special:corners}
	Let $r$ and $c$ be integer sequences of length $m$ and $n$, respectively, and let $M$ be the merge matrix of $r$ and $c$. 
	Let $i \in \argmin(r)$ and $j \in \argmin(c)$. 
	Let furthermore $k \in \argmin(r[\{1, \ldots, m\} \setminus \{i\}])$
		and $\ell \in \argmin(c[\{1, \ldots, n\} \setminus \{i\})$
	Let $\{p^*, q^*\} = \{\rightup{p}, \upright{p}\}$ such that $\max(M[p^*]) \le \max(M[q^*])$.
	\begin{enumerate}
		\item\label{lem:special:corners:topright} If $i = m$, $j = n$, $k = 1$, and $\ell = 1$, then $p^* \dominates \allpaths(M)$.
		\item\label{lem:special:corners:bottomleft} If $i = 1$, $j = 1$, $k = m$, and $\ell = n$, then $p^* \dominates \allpaths(M)$.
	\end{enumerate}
\end{lemma}
\begin{proof}
	\cref{lem:special:corners:topright}. First, we may assume that $r(1) > r(m)$ and that $c(1) > c(n)$, otherwise we could have applied one of the cases of \cref{lem:special}. We prove the lemma in two steps:
	\begin{enumerate}[label={\arabic*.},ref={\arabic*}]
		\item\label{corners:step1} We show that for each path $q$ in $M$, $\rightup{p}$ or $\upright{p}$ (or both) dominate(s) $q$.
		\item\label{corners:step2} We show that $\rightup{p}$ dominates $\upright{p}$, or vice versa, or both (depending on which case we are in).
	\end{enumerate}
\begin{figure}
	\centering
		\includegraphics[height=.3\textheight]{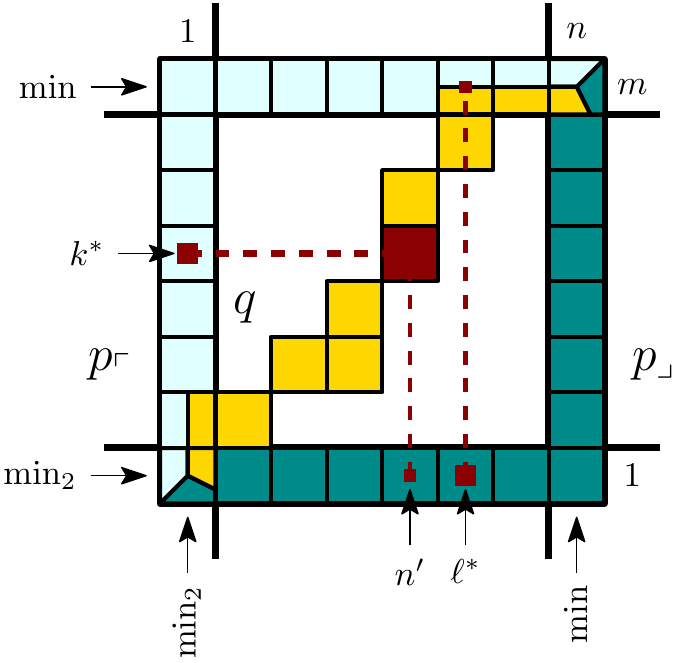}
		\caption{Situation in the first stage of the proof of \cref{lem:special:corners}\cref{lem:special:corners:topright}.
			The row and column labeled `$\min$' contains the minimum element from the respective sequence, 
			and the row and column labeled `$\kmin{2}$' contains the minimum element among all elements except 
			the one in the $\min$-row or -column.}
		\label{fig:corners}
\end{figure}
	The following claim will be useful in both steps and can be seen as a slight generalization of \cref{lem:special}.
	\begin{nestedclaim}\label{claim:corners:max}
		Let $q \in \allpaths(M)$ and let $p \in \{\rightup{p}, \upright{p}\}$. If $\max(M[p]) \le \max(M[q])$, then $p \dominates q$.
	\end{nestedclaim}
	\begin{claimproof}
		Suppose that $p = \rightup{p}$, the other case is symmetric. The claim can be shown using the same argument as in \cref{lem:special}, paying slight attention to the situation in which the maximum value of $q$ is in row $m$, which implies that the maximum of $\rightup{p}$ is in the same column.
	\end{claimproof}
	We prove Step \cref{corners:step1}. For the following argument, see \cref{fig:corners}. If $\max(M[q]) \ge \max(M[\rightup{p}])$, then we conclude by \cref{claim:corners:max} that $\rightup{p} \dominates q$ and we are done with Step \cref{corners:step1} of the proof. Suppose 
	\begin{align}\label{eq:max}
		\max(M[q]) < \max(M[\rightup{p}])
	\end{align}
	and let $\ell^* \in \argmax(M[\rightup{p}])$. 
	We may assume that $\ell^* < n$: otherwise, \eqref{eq:max} cannot be satisfied since $n \in \argmin(c)$.
	We furthermore have that $q$ contains $(m, \ell^*)$, since $\rightup{p}$ contains $(1, \ell^*)$, and
	$m$ is the only position in which $r$ is (potentially) smaller than $r(1)$.
	Therefore, this is the only way in which \eqref{eq:max} can be satisfied. 
	
	Now let $k^* \in \argmax(M[\upright{p}])$. As above, we may assume that $k^* < m$. Now, since $q$ contains $(m, \ell^*)$, we have that $q$ also contains $(k^*, n')$ for some $n' < n$. It follows that
	\begin{align*}
		\max(M[\upright{p}]) = M[k^*, 1] \le M[k^*, n'] \le \max(M[q])
	\end{align*}
	where the first inequality follows from the fact that $r(1) \le r(n')$ for all $n' < n$. By \cref{claim:corners:max}, $\upright{p}$ dominates $q$ and we finished Step \cref{corners:step1} of the proof. Step \cref{corners:step2} follows from another application of \cref{claim:corners:max} and the lemma follows from transitivity of the domination relation. This proves \cref{lem:special:corners:topright}, and \cref{lem:special:corners:bottomleft} follows from a symmetric argument.
\end{proof}

\begin{remark}
	We would like to stress that up to this point, all results in this section were shown in terms of arbitrary integer sequences. For the next lemma, we require the sequences considered to be \emph{typical sequences}. In \cref{sec:split-chop:arbitrary} we will generalize the results that rely on the following lemmas to arbitrary integer sequences.
\end{remark}

We are now ready to prove the Chop Lemmas. They come in two versions, one that is suited for the case of the bottom left submatrix after an application of the Split Lemma to $M$, and one for the top right submatrix. In the former case, we have that the last row is a minimum row and that the last column is a minimum column. We will prove this lemma in more detail and observe that the other case follows by symmetry with the arguments given in the following proof.
For an illustration of the setting in the following lemma, see \cref{fig:chop-general}.
\begin{figure}
	\begin{subfigure}[t]{.4\textwidth}
		\centering
		\includegraphics[height=.18\textheight]{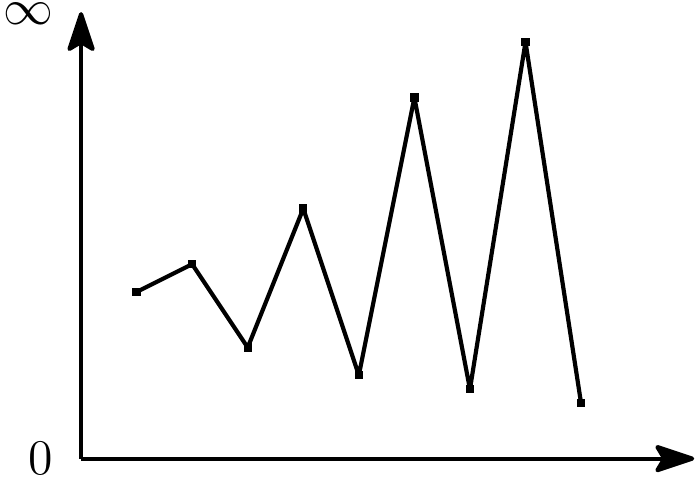}
		\caption{Typical sequence ending in the minimum.}
		\label{fig:typseq:end:min}
	\end{subfigure}
	\hfill
	\begin{subfigure}[t]{.58\textwidth}
		\centering
		\includegraphics[height=.32\textheight]{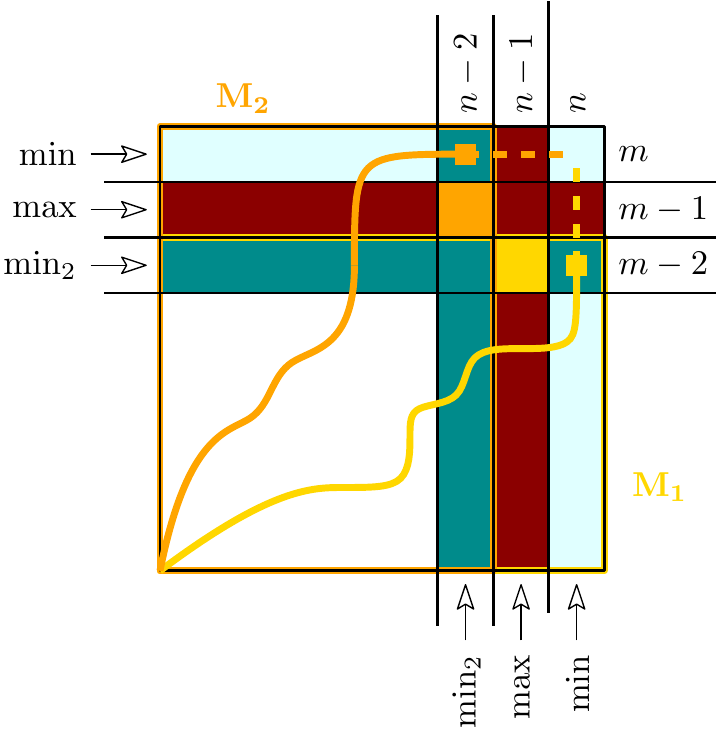}
		\caption{The basic setup in \cref{lem:chop}.}
		\label{fig:chop-general}
	\end{subfigure}
	\caption{Visual aides to the proof of \cref{lem:chop}.}
\end{figure}
\begin{lemma}[Chop Lemma - Bottom]\label{lem:chop}
	Let $r$ and $c$ be typical sequences of length $m \ge 3$ and $n \ge 3$, respectively, 
	and let $M$ be the merge matrix of $r$ and $c$. 
	Suppose that $m \in \argmin(r)$ and $n \in \argmin(c)$ 
	and let $M_1 \defeq M[1..(m-2),1..n]$ and $M_2 \defeq M[1..m,1..(n-2)]$ 
	and for all $h \in [2]$, let $p_h \dominates \allpaths(M_h)$.
	Let $p_1^+ \defeq p_1 \concat (m-1, n), (m, n)$ and $p_2^+ \defeq p_2 \concat (m, n-1), (m, n)$.
	\begin{enumerate}
		\item\label{lem:chop:1} If $M[m-2, n-1] \le M[m-1, n-2]$, then $p_1^+ \dominates \allpaths(M)$.
		\item\label{lem:chop:2} If $M[m-1, n-2] \le M[m-2,n-1]$, then $p_2^+ \dominates \allpaths(M)$.
	\end{enumerate}
\end{lemma}
\begin{proof}
	Let $s \in \{r, c\}$. Since $s$ is a typical sequence and $\length(s) \in \argmin(s)$, 
	we know by \cref{cor:typseq:min:shape} that for all
	$k \in \left[\left\lfloor \length(s)/2 \right\rfloor\right]$,
	\begin{align*}
		\length(s) - 2k + 1 \in \argmax(s[1..(\length(s)-2k+1)]) \mbox{ and } \length(s) - 2k \in \argmin(s[1..(\length(s)-2k)]).
	\end{align*}
	
	Informally speaking, this means that the last element of $s$ is the minimum, the $(\length(s)-1)$-th element of $s$ is the maximum, the $(\length(s)-2)$-th element is `second-smallest' element, and so on. 
	We will therefore refer to the element at position $\length(s)-2k$ ($2k \le \length(s)$) as `$\kmin{k+1}(s)$' 
	(note that the minimum is achieved when $k = 0$, hence the `$+1$'), 
	and elements at position $\length(s)-2k+1$ ($2k+1 \le \length(s)-1$) as `$\kmax{k}(s)$'.
	For an illustration of the shape of $s$ see \cref{fig:typseq:end:min} and 
	for an illustration of the basic setting of this proof see \cref{fig:chop-general}. 
	We prove~\cref{lem:chop:1} and remark that the argument for~\cref{lem:chop:2} is symmetric.
	
	First, we show that each path in $M$ is dominated by at least one of $p_1^+$ and $p_2^+$.
	\begin{nestedclaim}\label{obs:p1star:p2star}
		Let $q \in \allpaths(M)$. Then, for some $r \in [2]$, $p_r^+ \dominates q$.
	\end{nestedclaim}
	\begin{claimproof}
		We may assume that $q$ does not contain $(m-1, n-1)$: if so, we could easily obtain a path $q'$ from $q$ by some local replacements such that $q'$ dominates $q$, since $M[m-1, n-1]$ is the maximum element of the matrix $M$. We may assume that $q$ either contains $(m-1, n)$ or $(m, n - 1)$. Assume that the former holds, and note that an argument for the latter case can be given analogously. 
		Since $q$ contains $(m-1, n)$, and since $q$ does not contain $(m-1, n-1)$, we may assume that $q$ contains $(m-2, n)$: if not, we can simply add $(m-2, n)$ before $(m-1, n)$ to obtain a path that dominates $q$ (recall that $n$ is the column indexed by the minimum of $c$). Now, let $q|_{M_1}$ be the restriction of $q$ to $M_1$, we then have that $q = q|_{M_1} \concat (m-1, n), (m, n)$. Since $p_1$ dominates all paths in $M_1$, it dominates $q|_{M_1}$ and so $p_1^+ \dominates q$.
	\end{claimproof}
	
	The remainder of the proof is devoted to showing that $p_1^+$ dominates $p_2^+$ which yields the lemma by \cref{obs:p1star:p2star} and transitivity. To achieve that, we will show in a series of claims that we may assume that 
	$p_2$ contains $(m-2, n-2)$.
	In particular, we show that if $p_2$ does not contain $(m-2, n-2)$, 
	then there is another path in $M_2$ that does contain $(m-2,n-2)$ and dominates $p_2$. 
	\begin{nestedclaim}\label{claim:chop:unique:max}
		We may assume that there is a unique $j \in [n-2]$ such that $p_2$ contains $(m-1, j)$.
	\end{nestedclaim}
	\begin{claimproof}
		Clearly, $p_2$ has to pass through the row $m-1$ at some point. We show that we may assume that there is a unique such point.
		Suppose not and let $j_1, \ldots, j_t$ be such that $p_2$ contains all $(m-1, j_i)$, where $i \in [t]$. 
		By the definition of a path in a matrix, we have that $j_{i+1} = j_i + 1$ for all $i \in [t-1]$. 
		Let $p_2'$ be the path obtained from $p_2$ by replacing, for each $i \in [t-1]$, the element $(m-1, j_i)$ with the element $(m-2, j_i)$. Since $r(m-2) \le r(m-1)$ (recall that $m-1 \in \argmax(r)$), it is not difficult to see that $p_2'$ dominates $p_2$, and clearly, $p_2'$ satisfies the condition of the claim.
	\end{claimproof}
	\begin{figure}
		\centering
		\begin{subfigure}[t]{.48\textwidth}
			\includegraphics[width=.9\textwidth]{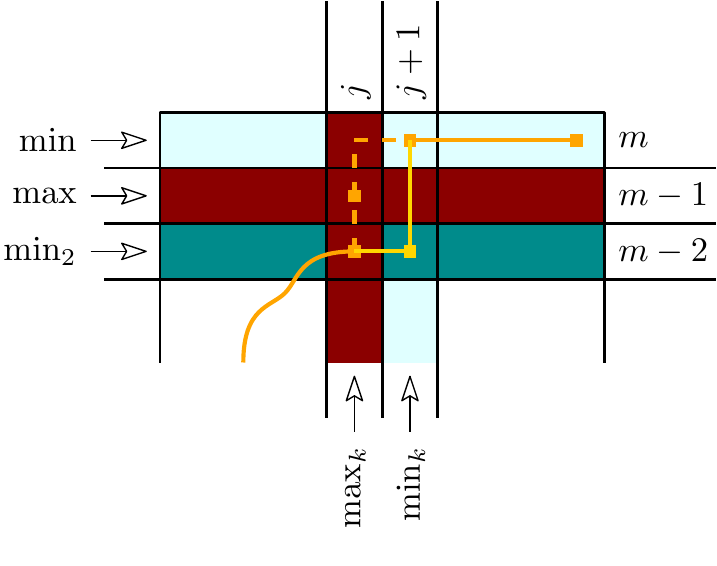}
			\caption{Situation of \cref{claim:chop:push:max}.}
			\label{fig:chop:shift:max}
		\end{subfigure}
		\hfill
		\begin{subfigure}[t]{.48\textwidth}
			\includegraphics[width=.9\textwidth]{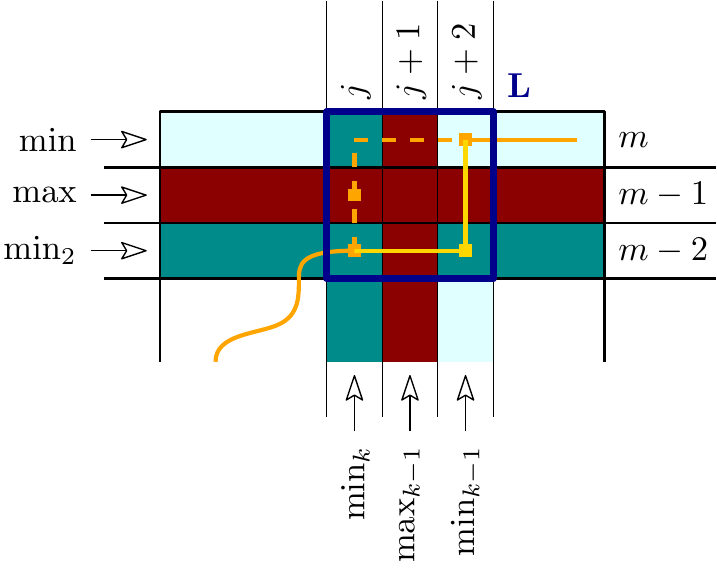}
			\caption{Situation of \cref{claim:chop:push:min}.}
			\label{fig:chop:shift:min}
		\end{subfigure}
		\caption{Visualization of the arguments that lead to the conclusion that we may assume that $p_2$ contains $(m-2, n-2)$ in the proof of \cref{lem:chop}.}
	\end{figure}
	\begin{nestedclaim}\label{claim:chop:push:max}
		Let $j \in [n-3]$ be such that $p_2$ contains $(m-1, j)$. 
		If $j = n - 2k + 1$ for some $k \in \bN$ with $2k + 1 \le n - 1$, 
		then there is a path $p_2'$ that dominates $p_2$ and contains $(m-1, j + 1)$.
	\end{nestedclaim}	
	\begin{claimproof}
		For an illustration see \cref{fig:chop:shift:max}.
		First, by \cref{claim:chop:unique:max}, we may assume that $j$ is unique.
		Moreover, since $j = n - 2k + 1$ and $j + 1 = n - 2k + 2 = n - 2(k-1)$, 
		we have that $c(j) = \kmax{k}(c)$ and $c(j+1) = \kmin{k}(c)$, respectively,
		and therefore $c(j+1) \le c(j)$.
		Hence, we may assume that the element after $(m-1, j)$ in $p_2$ is $(m, j+1)$: 
		if $p_2$ contained $(m, j)$ we could simply remove $(m, j)$ from $p_2$ 
		without changing the fact that $p_2$ is a dominating path since $M[m, j] > M[m, j + 1]$. 
		We modify $p_2$ as follows. We remove $(m-1, j)$, and add $(m-2, j)$ (if not already present), 
		followed by $(m-2, j+1)$ and then $(m-1, j+1)$. 
		For each $x \in \{M[m-2, j], M[m-2, j+1], M[m-1, j+1]\}$, we have that $x < M[m-1, j]$ 
		(recall that $r(m-2) < r(m-1)$ and $c(j + 1) < c(j)$).
		Hence, the resulting path dominates $p_2$ and it contains $(m-1, j+1)$.
	\end{claimproof}
	
	\begin{nestedclaim}\label{claim:chop:push:min}
		Let $j \in [n-4]$ be such that $p_2$ contains $(m-1, j)$. If $j = n - 2(k - 1)$ for some 
		$k \in \left[3..\left\lfloor \frac{n}{2}\right\rfloor\right]$, 
		then there is a path $p_2'$ that dominates $p_2$ and contains $(m-1, j + 2)$.
	\end{nestedclaim}
	\begin{claimproof}
		For an illustration see \cref{fig:chop:shift:min}.
		Again, by \cref{claim:chop:unique:max}, we may assume that $j$ is unique.
		Since $j = n - 2(k-1)$, we have that $c(j) = \kmin{k}(c)$.
		First, if not already present, 
		we insert $(m-2, j)$ just before $(m-1, j)$ in $p_2$. 
		This does not change the fact that $p_2$ is a dominating path, 
		since $M[m-2, j] < M[m-1, j]$ (recall that $r(m-2) < r(m-1)$). 
		Next, consider the $3 \times 3$ submatrix $L \defeq M[(m-2)..m,j..(j+2)]$. 
		Note that $L$ is the submatrix of $M$ restricted to the rows 
		$\min(r)$, $\max(r)$, and $\kmin{2}(r)$, and the columns 
		$\kmin{k}(c)$, $\kmax{k - 1}(c)$, and $\kmin{k-1}(c)$.
		Furthermore, we have that $p_2$ restricted to $L$ is equal to $\upright{p}(L)$. 
		We show that $\rightup{p}(L)$ dominates $\upright{p}(L)$, 
		from which we can conclude that we can obtain a path $p_2'$ from $p_2$ 
		that contains $(m-1, j+2)$ and dominates $p_2$
		by replacing $\upright{p}(L)$ with $\rightup{p}(L)$. 
		By \cref{lem:special:corners}, it suffices to show that $M[m-2, j+1] \le M[m-1, j]$, 
		in other words, that $\kmax{k-1}(c) + \kmin{2}(r) \le \max(r) + \kmin{k}(c)$.
		
		By the assumption of the lemma, we have that $M[m-2, n-1] \le M[m-1, n-2]$, hence,
		\begin{align*}
			\max(c) + \kmin{2}(r) \le \max(r) + \kmin{2}(c), \mbox{ and so: } \max(c) - \kmin{2}(c) \le \max(r) - \kmin{2}(r).
		\end{align*}
		Next, we have that for all $j \in \left[\left\lfloor n/2 \right\rfloor\right]$,
	\begin{align*}
		\max(c) - \kmin{2}(c) \ge \kmax{j}(c) - \kmin{j+1}(c).
	\end{align*}
	Putting the two together, we have that 
	\begin{align*}
		\kmax{k-1}(c) - \kmin{k}(c) \le \max(r) - \kmin{2}(r), \mbox{ and so: } \kmax{k-1}(c) + \kmin{2}(r) \le \max(r) + \kmin{k}(c),
	\end{align*}
	which concludes the proof of the claim.
	\end{claimproof}
	
	We are now ready to conclude the proof.
	\begin{nestedclaim}
		$p_1^+ \dominates p_2^+$.
	\end{nestedclaim}
	\begin{claimproof}
		By repeated application of \cref{claim:chop:push:max,claim:chop:push:min}, 
		we know that there is a path $p_2'$ in $M_2$ that contains $(m-1, n-2)$.
		Furthermore, we may assume that $p_2'$ contains $(m-2, n-2)$ as well:
		we can simply add this element if it is not already present; since $M[m-2, n-2] \le M[m-1, n-2]$,
		this does not change the property that $p_2' \dominates p_2$.
		Now, let $p_2''$ be the subpath of $p_2'$ ending in $(m-2, n-2)$.
		(Note that $p_2'' \circ (m-2, n-1), (m-2, n) \in \allpaths(M_1)$.)
		Then, 
		\begin{align}
			p_1^+ &\dominates p_2'' \circ (m-2, n-1), (m-2, n), (m-1, n), (m, n) \label{eq:chop:final:1} \\
				&\dominates p_2' \circ (m, n-1), (m, n) \label{eq:chop:final:2} \\
				&\dominates p_2^+, \label{eq:chop:final:3}
		\end{align}
		where~\eqref{eq:chop:final:1} is due to $p_1 \dominates \allpaths(M_1)$ and therefore $p_1 \dominates p_2'' \circ (m-2, n-1), (m-2, n)$. 
		Next~\eqref{eq:chop:final:2} follows from an application of \cref{lem:special:corners} to 
		the $3 \times 3$-submatrix $M[(m-2)..m,(n-2)..n]$ 
		and~\eqref{eq:chop:final:3} is guaranteed since $p_2' \dominates p_2$.
	\end{claimproof}
	This concludes the proof of \cref{lem:chop:1} and \cref{lem:chop:2} can be shown symmetrically.
\end{proof}

As the previous lemma always assumes that $m \ge 3$ and $n \ge 3$, 
we observe the corresponding base case which occurs when either $m \le 2$ or $n \le 2$.
This base case is justified by the observation that in the bottom case, 
the last row and column of $M$ are minimum.
\begin{observation}[Base Case - Bottom]\label{obs:base:bottom}
	Let $r$ and $c$ be typical sequences of length $m$ and $n$, respectively, and let $M$ be the merge matrix of $r$ and $c$. 
	Suppose that $m \in \argmin(r)$ and $n \in \argmin(c)$.
	If $m \le 2$ ($n \le 2$), then\footnote{Note that in the following equation, if $m = 1$, then strictly speaking we would have that $p^*$ repeats the element $(1, 1)$ twice which is of course not our intention. For the sake of a clear presentation though, we will ignore this slight abuse of notation, also in similar instances throughout this section.}
	\[
		p^* \defeq (1,1), (m, 1), (m, 2), \ldots, (m, n) ~~~ (p^* \defeq (1, 1), (1, n), (2, n), \ldots, (m, n))
	\]
 	dominates $\allpaths(M)$, i.e.\ $p^* \dominates \allpaths(M)$.
\end{observation}

By symmetry, we have the following consequence of \cref{lem:chop}.
\begin{corollary-bf}[Chop Lemma - Top]\label{lem:chop:top}
	Let $r$ and $c$ be typical sequences of length $m \ge 3$ and $n \ge 3$, respectively, and let $M$ be the merge matrix of $r$ and $c$. 
	Suppose that $1 \in \argmin(r)$ and $1 \in \argmin(c)$ and let $M_1 \defeq M[3..m,1..n]$ and $M_2 \defeq M[1..m,3..n]$ and for all $h \in [2]$, let $p_h \dominates \allpaths(M_h)$. Let $p_1^+ \defeq (1, 1), (2, 1) \concat p_1$ and $p_2^+ \defeq (1, 1), (1, 2) \concat p_2$.
	\begin{enumerate}
		\item\label{lem:chop:top:1} If $M[3, 2] \le M[2, 3]$, then $p_1^+ \dominates \allpaths(M)$.
		\item\label{lem:chop:top:2} If $M[2, 3] \le M[3, 2]$, then $p_2^+ \dominates \allpaths(M)$.
	\end{enumerate}
\end{corollary-bf}

Again, we observe the corresponding base case.
\begin{observation}[Base Case - Top]\label{obs:base:top}
	Let $r$ and $c$ be typical sequences of length $m$ and $n$, respectively, and let $M$ be the merge matrix of $r$ and $c$. 
	Suppose that $1 \in \argmin(r)$ and $1 \in \argmin(c)$.
	If $m \le 2$ ($n \le 2$), then
	\[
		p^* \defeq (1,1), (1, 2), \ldots, (1, n), (m, n) ~~~ (p^* \defeq (1, 1), (2, 1), \ldots, (m, 1), (m, n))
	\]
	dominates $\allpaths(M)$, i.e.\ $p^* \dominates \allpaths(M)$.
\end{observation}

\subsection{The Split-and-Chop Algorithm}\label{sec:alg:split-and-chop}
Equipped with the Split Lemma and the Chop Lemmas, we are now ready to give the algorithm that computes a dominating merge of two typical sequences. Consequently, we call this algorithm the `Split-and-Chop Algorithm'.

\begin{algorithm}[t]
	\SetKwInOut{algin}{Input}
	\SetKwInOut{algout}{Output}
	\algin{Typical sequences $r(1), \ldots, r(m)$ and $c(1), \ldots, c(n)$}
	\algout{A dominating merge of $r$ and $c$}
	\SetKwFunction{chopbtm}{Chop-bottom}
	\SetKwFunction{choptop}{Chop-top}
	Let $i \in \argmin(r)$ and $j \in \argmin(c)$\;
	\KwRet \chopbtm($r[1..i]$, $c[1..j]$)${}\concat{}$\choptop($r[i..m], c[j..n]$)\;
  \SetKwProg{myproc}{Procedure}{}{}
  \myproc{\chopbtm{$r$ and $c$ as above}}{
	\lIf{$m \le 2$}{\KwRet $r(1) + c(1)$, $r(m) + c(1)$, $r(m) + c(2)$, $\ldots$, $r(m) + c(n)$}
	\lIf{$n \le 2$}{\KwRet $r(1) + c(1)$, $r(1) + c(n)$, $r(2) + c(n)$, $\ldots$, $r(m) + c(n)$}
  \lIf{$r(m-2) + c(n-1) \le r(m-1) + c(n-2)$}{\KwRet \chopbtm{$r[1..(m-2)], c$}${}\concat{} (r(m-1) + c(n)), r(m) + c(n)$}
  \lIf{$r(m-1) + c(n-2) \le r(m-2) + c(n-1)$}{\KwRet \chopbtm{$r, c[1..(n-2)]$}${}\concat{} (r(m) + c(n-1)), r(m) + c(n)$}
  }
  \SetKwProg{myproc}{Procedure}{}{}
  \myproc{\choptop{$r$ and $c$ as above}}{
  \lIf{$m \le 2$}{\KwRet $r(1) + c(1)$, $r(1) + c(2)$, $\ldots$, $r(1) + c(n)$, $r(m) + c(n)$}
  \lIf{$n \le 2$}{\KwRet $r(1) + c(1)$, $r(2) + c(1)$, $\ldots$, $r(m) + c(1)$, $r(m) + c(n)$}
  \lIf{$r(3) + c(2) \le r(2) + c(3)$}{\KwRet $r(1) + c(1), (r(2) + c(1)){}\concat{}$\choptop{$r[3..m], c$}}
  \lIf{$r(2) + c(3) \le r(3) + c(2)$}{\KwRet $r(1) + c(1), (r(1) + c(2)){}\concat{}$\choptop{$r, c[3..n]$}}
  }
  \caption{The Split-and-Chop Algorithm}
  \label{alg:split-and-chop}
\end{algorithm} 

\begin{lemma}\label{lem:split-and-chop:typical}
	Let $r$ and $c$ be typical sequences of length $m$ and $n$, respectively. 
	Then, there is an algorithm that finds in $\cO(m + n)$ time a dominating path in the merge matrix of $r$ and $c$.
\end{lemma}
\begin{proof}
	The algorithm practically derives itself from the Split Lemma (\cref{lem:split}) and the Chop Lemmas (\cref{lem:chop,lem:chop:top}). However, to make the algorithm run in the claimed time bound, we are not able to construct the merge matrix of $r$ and $c$. This turns out to be not necessary, as we can simply read off the crucial values upon which the recursion of the algorithm depends from the sequences directly. The details are given in \cref{alg:split-and-chop}.
	
	The runtime of the Chop-subroutines can be computed as $T(m + n) \le T(m + n - 2) + \cO(1)$, which resolves to $\cO(m + n)$. Correctness follows from \cref{lem:split,lem:chop,lem:chop:top} with the base cases given in \cref{obs:base:bottom,obs:base:top}.
\end{proof}

\subsection{Generalization to Arbitrary Integer Sequences}\label{sec:split-chop:arbitrary}

In this section we show how to generalize \cref{lem:split-and-chop:typical} to arbitrary integer sequences. In particular, we will show how to construct from a merge of two typical sequences $\typseq(r)$ and $\typseq(s)$ that dominates all of their merges, a merge of $r$ and $s$ that dominates all merges of $r$ and $s$. The claimed result then follows from an application of \cref{lem:split-and-chop:typical}.
We illustrate the following construction in \cref{fig:typical:lift}. 

\paragraph{The Typical Lift.} Let $r$ and $s$ be integer sequences and let $t \in \typseq(r) \allmerges \typseq(s)$. Then, the \emph{typical lift of $t$}, denoted by $\untypseq(t)$, is an integer sequence $\untypseq(t) \in r \allmerges s$, obtained from $t$ as follows. 
	For convenience, we will consider $\untypseq(t)$ as a path in the merge matrix $M$ of $r$ and $s$.
\begin{description}
	\item[Step 1.] We construct $t' \in \typseq(r) \allmergesnd \typseq(s)$ such that $t' \dominates t$ using Lemma~\ref{lem:merge:nd:equiv}.
		Throughout the following, consider $t'$ to be a path in the merge matrix $M_\typseq$ of $\typseq(r)$ and $\typseq(s)$.
	\item[Step 2.] 
		First, we initialize $\untypseq_t^1 \defeq t'(1) = (1, 1)$.
		For $i = \{2, \ldots, \length(t')\}$, we proceed inductively as follows.
		Let $(i_r, i_s) = t(i)$ and let $(i'_r, i'_s) = t(i-1)$.
		(Note that $t(i-1)$ and $t(i)$ are indices in $M_\typseq$.)
		Let furthermore $(j_r, j_s)$ be the index in $M$ corresponding to $(i_r, i_s)$, 
		and let $(j'_r, j'_s)$ be the index in $M$ corresponding to $(i'_r, i'_s)$.
		Assume by induction that $\untypseq_t^{i-1} \in \allpaths(M[1..j_r', 1..j_s'])$. 
		We show how to extend $\untypseq_t^{i-1}$ to a path in $\untypseq_t^i$
		in $M[1..j_r, 1..j_s]$.
		Since $t'$ is non-diagonal, we have that $(i_r', i_s') \in \{(i_r - 1, i_s), (i_r, i_s - 1)\}$, 
		so one of the two following cases applies.
	\item[Case S2.1 ($i_r' = i_r - 1$ and $i_s' = i_s$).] In this case, we let
			$\untypseq_t^i \defeq \untypseq_t^{i-1} \concat (j_r' + 1, j_s), \ldots, (j_r, j_s)$.
	\item[Case S2.2 ($i_r' = i_r$ and $i_s' = i_s - 1$).]	In this case, we let
			$\untypseq_t^i \defeq \untypseq_t^{i-1} \concat (j_r, j_s' + 1), \ldots, (j_r, j_s)$.
	\item[Step 3.] We return $\untypseq(t) \defeq \untypseq_t^{\length(t')}$.
\end{description}
\begin{figure}
	\begin{subfigure}[t]{.26\textwidth}
		\centering
		\includegraphics[width=.95\textwidth]{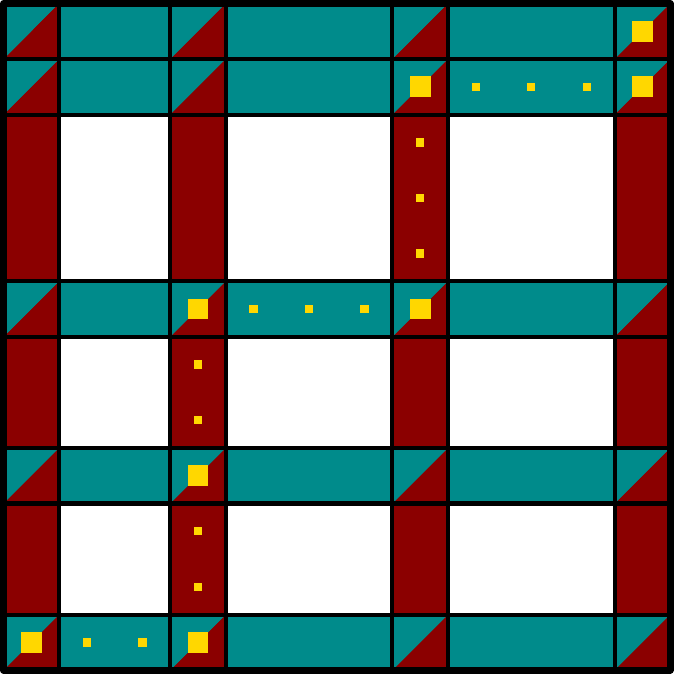}
	\end{subfigure}
	\hfill
	\begin{subfigure}[t]{.72\textwidth}
		\centering
		\includegraphics[width=.95\textwidth]{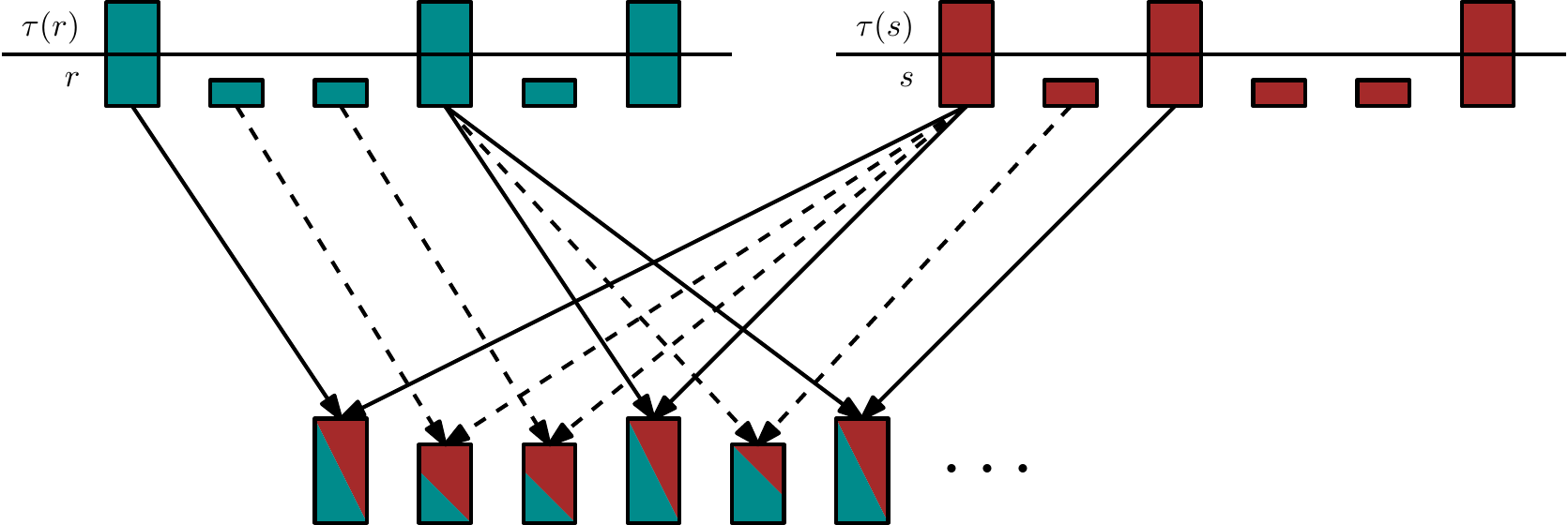}
	\end{subfigure}
	\caption{Illustration of the typical lift. 
		On the left side, the view of the merge matrix $M$, with the rows and columns corresponding to elements of the typical sequences highlighted.
		Inside there, $M_\typseq$ can be seen as a highlighted submatrix.
		The merge $t'$ is depicted as the large yellow squares within $M_\typseq$ 
		and the small yellow squares outside of $M_\typseq$ show its completion to the typical lift of $t$.
		On the right side, an illustration that does not rely on the `matrix view'.}
	\label{fig:typical:lift}
\end{figure}

Furthermore, it is readily seen that 
the typical lift contains no diagonal steps: we obtain it from a non-diagonal path 
in the merge matrix of $\typseq(r)$ and $\typseq(s)$ 
by inserting vertical and horizontal paths from the merge matrix of $r$ and $s$
between consecutive elements.
Moreover, it is computable in linear time, with Step 1 taking linear time by \cref{lem:merge:nd:equiv}. 
We summarize in the following observation.
\begin{observation}\label{obs:typical:lift:basics}
	Let $r$ and $s$ be integer sequences of length $m$ and $n$, respectively, and let $t \in \typseq(r) \allmerges \typseq(s)$. 
	Then, $\untypseq(t) \in r \allmergesnd s$, and $\untypseq(t)$ can be computed in time $\cO(m + n)$.
\end{observation} 

We now show that if $t \in \typseq(r) \allmerges \typseq(s)$ dominates all merges of $\typseq(r)$ and $\typseq(s)$, 
then the typical lift of $t$ dominates all merges of $r$ and $s$. 

\begin{lemma}\label{lem:typical:lift}
	Let $r$ and $s$ be integer sequences and let $q \in r \allmerges s$. 
	Let $t \in \typseq(r) \allmerges \typseq(s)$ such that $t \dominates \typseq(r) \allmerges \typseq(s)$. 
	Then, $\untypseq(t) \dominates q$.
\end{lemma}
\begin{proof}
	Let $t' \in \typseq(r) \allmergesnd \typseq(s)$ be the non-diagonal merge such that $t' \dominates t$
	used in the construction of $\untypseq(t)$.
	We argue that $\untypseq(t) \dominates t'$.
	To see this, let $M$ be the merge matrix of $r$ and $s$ 
	and consider any $(j_r', j_s')$ and $(j_r, j_s)$ as in Step 2, and suppose that $j_s' = j_s$.
	(Note that either $j_s' = j_s$ or $j_r' = j_r$.)
	As the only elements of the typical sequence of $r$ in $[j_r'..j_r]$ are $r(j_r')$ and $r(j_r)$, 
	we know that either for all $h_r \in [j_r'..j_r]$, $r(j_r') \le r(h_r) \le r(j_r)$,
	or for all $h_r \in [j_r'..j_r]$, $r(j_r') \ge r(h_r) \ge r(j_r)$.
	Therefore, in an extension of $t'$, we can repeat the index that yields $\max\{M[j_r', j_s], M[j_r, j_s]\}$
	sufficiently many (i.e.\ $j_r - j_r'$) times to ensure that the value of the extension of $t'$ 
	is an upper bound for all values of $\untypseq(t)$ in these positions.
	
	To finish the proof, we have by \cref{lem:BK}\cref{lem:BK:3:14} that there exists a $q' \in \typseq(r) \allmerges \typseq(s)$
	such that $q' \dominates q$. Since $t \dominates \typseq(r) \allmerges \typseq(s)$, we can conclude:
	\begin{align*}
		\untypseq(t) \dominates t' \dominates t \dominates q' \dominates q.
	\end{align*}
\end{proof}
We wrap up and prove the Merge Dominator Lemma (\cref{lem:merge:dom}), 
stated here in the slightly stronger form that the dominating merge is non-diagonal 
(which is necessary for the applications in \cref{sec:width:measures}).
\begin{lemma}[Merge Dominator Lemma]\label{lem:merge:dom:strong}
	Let $r$ and $c$ be integer sequence of length $m$ and $n$, respectively. There exists a dominating non-diagonal merge of $r$ and $c$, i.e.\ an integer sequence $t \in r \allmergesnd c$ such that $t \dominates r \allmerges c$, and this dominating merge can be computed in time $\cO(m + n)$.
\end{lemma}
\begin{proof}
	The algorithm proceeds in the following steps.
	\begin{description}
		\item[Step 1.] Compute $\typseq(r)$ and $\typseq(c)$.
		\item[Step 2.] Apply the Split-and-Chop Algorithm on input ($\typseq(r)$, $\typseq(c)$) to obtain $t \dominates \typseq(r) \allmerges \typseq(c)$.
		\item[Step 3.] Return the typical lift $\untypseq(t)$ of $t$.
	\end{description}
	Correctness of the above algorithm follows from \cref{cor:merge:path,lem:split-and-chop:typical,lem:typical:lift} which together guarantee that $\untypseq(t) \dominates r \allmerges c$, and by \cref{obs:typical:lift:basics}, $\untypseq(t)$ is a non-diagonal merge, i.e.\ $\untypseq(t) \in r \allmergesnd c$. By \cref{lem:compute:typseq}, Step 1 can be done in time $\cO(m + n)$, by \cref{lem:split-and-chop:typical}, Step 2 takes time $\cO(m + n)$ as well, and by \cref{obs:typical:lift:basics}, the typical lift of $t$ can also be computed in time $\cO(m + n)$. Hence, the overall runtime of the algorithm is $\cO(m + n)$.
\end{proof}

\section{Directed Width Measures of Series Parallel Digraphs}\label{sec:width:measures}
In this section, we give algorithmic consequences of the Merge Dominator Lemma.
In \cref{sec:cutw}, we provide an algorithm that
computes the (weighted) cutwidth of (arc-weighted) series parallel digraphs on $n$ vertices in time $\cO(n^2)$.
In \cref{sec:mcutw} we provide a linear-time transformation
that allows for computing the modified cutwidth of an SPD on $n$ vertices in $\cO(n^2)$ time,
using the algorithm that computes the weighted cutwidth of an arc-weighted SPD.

\subsection{Cutwidth}\label{sec:cutw}
Recall that given a topological order $v_1, \ldots, v_n$ of a directed acyclic graph $G$, 
its cutwidth is the maximum over all $i \in [n-1]$ of the number of arcs that have 
their tail vertex in $\{v_1, \ldots, v_i\}$ and their head vertex in $\{v_{i+1}, \ldots, v_n\}$, 
and that the cutwidth of $G$ is the minimum cutwidth over all its topological orders.
We will now deal with the following computational problem.

\fancyproblemdef
	{Cutwidth of Series Parallel Digraphs}
	{A series parallel digraph $G$.}
	{What is the cutwidth of $G$?}
	
Given a series parallel digraph $G$, we follow a bottom-up dynamic programming scheme 
along the decomposition tree $T$ that yields $G$. 
Each node $t \in V(T)$ has a subgraph $G_t$ of $G$ associated with it, that is also series parallel.
Naturally, we use the property that $G_t$ is obtained either via series or parallel
composition of the SPD's associated with its two children.

To make this problem amenable to be solved using merges of integer sequences, 
we define the following notion of a cut-size sequence of a topological order of a directed acyclic graph 
which records for each position in the order, how many arcs cross it.

\begin{definition}[Cut-Size Sequence]
	Let $G$ be a directed acyclic graph on $n$ vertices and let $\linord \in \toporders(G)$ be a topological order of $G$. 
	The sequence $x(1), \ldots, x(n-1)$, where for $i \in [n-1]$, 
	$$x(i) = \card{\{uv \in A(G) \mid \linord(u) \le i \wedge \linord(v) > i\}},$$ 
	is the \emph{cut-size sequence} of $\linord$, and denoted by $\cutsizeseq(\linord)$.
	For a set of topological orders $\toporders' \subseteq \toporders(G)$, we let $\cutsizeseq(\toporders') \defeq \{\cutsizeseq(\linord) \mid \linord \in \toporders'\}$.
\end{definition}

Throughout the remainder of this section, we slightly abuse notation: If $G_1$ and $G_2$ are SPD's that are being composed with a series composition, and $\linord_1 \in \toporders(G_1)$ and $\linord_2 \in \toporders(G_2)$, 
then we consider $\linord = \linord_1 \concat \linord_2$ to be the concatenation of the 
two topological orders where $t_2 = s_1$ only appears \emph{once} in $\linord$.

We first argue via two simple observations 
that when computing the cutwidth of a series parallel digraph $G$ by following its decomposition 
tree in a bottom up manner, we only have to keep track of a set of topological orders 
that induce a set of cut-size sequences that dominate all cut-size sequences of $G$. 
\begin{observation}\label{obs:cutsizeseq:dominate}
	Let $G$ be a DAG and $\linord, \linordd \in \toporders(G)$. 
	If $\cutsizeseq(\linord) \dominates \cutsizeseq(\linordd)$, then $\cutwidth(\linord) \le \cutwidth(\linordd)$.
\end{observation}

This is simply due to the fact that $\cutsizeseq(\linord) \dominates \cutsizeseq(\linordd)$ 
implies that $\max(\cutsizeseq(\linord)) \le \max(\cutsizeseq(\linordd))$.
Next, if $G$ is obtained from $G_1$ and $G_2$ via series or parallel composition, 
and we have $\linord_1, \linordd_1 \in \toporders(G_1)$ such that $\cutsizeseq(\linord_1) \dominates \cutsizeseq(\linordd_1)$,
then it is always beneficial to choose $\linord_1$ over $\linordd_1$, and $\linord_1$ can be disregarded.

\begin{observation}\label{lem:top:order:dominate}
	Let $G$ be an SPD that is obtained via series or parallel composition from SPD's $G_1$ and $G_2$.
	Let $\linord_1, \linordd_1 \in \toporders(G_1)$ be such that $\cutsizeseq(\linord_1) \dominates \cutsizeseq(\linordd_1)$.
	Let $\linord, \linordd \in \toporders(G)$ be such that $\linord|_{V(G_1)} = \linord_1$, $\linordd|_{V(G_1)} = \linordd_1$, 
	and for all $v \in V(G_2)$, $\linord(v) = \linordd(v)$.
	Then, $\cutsizeseq(\linord) \dominates \cutsizeseq(\linordd)$.
\end{observation}

The previous observation is justified as follows.
Let $\cutsizeseq(\linord) = x(1), \ldots, x(n-1)$ and $\cutsizeseq(\linordd) = y(1), \ldots, y(n-1)$.
Then, for each $i \in [n-1]$, 
the arcs of $G_2$ contribute equally to the values $x(i)$ and $y(i)$ 
(in particular since $G_1$ and $G_2$ are arc-disjoint).
Therefore, we can use extensions of $\cutsizeseq(\linord_1)$ and $\cutsizeseq(\linordd_1)$ 
that witnesses that $\cutsizeseq(\linord_1) \dominates \cutsizeseq(\linordd_1)$
to construct extensions of $\cutsizeseq(\linord)$ and $\cutsizeseq(\linordd)$ that witness that 
$\cutsizeseq(\linord) \dominates \cutsizeseq(\linordd)$.

The following lemma states that the cut-size sequences of an SPD $G$ can be computed by pairwise concatenation or non-diagonal merging 
(depending on whether $G$ is obtained via series or parallel composition) of the two smaller SPD's that $G$ is obtained from.
Intuitively speaking, the reason why we can only consider \emph{non-diagonal} merges is the following.
When $G$ is obtained from $G_1$ and $G_2$ via parallel composition, then each topological order
of $G$ can be considered the `merge' of a topological order of $G_1$ and one of $G_2$, where 
each position (apart from the first and the last) contains a vertex \emph{either} from $G_1$ \emph{or} from $G_2$.
Now, in a merge of a cut-size sequence of $G_1$ with a cut-size sequence of $G_2$, 
a diagonal step would essentially mean that in some position, 
we insert both a vertex from $G_1$ and a vertex of $G_2$; this is of course not possible.
\begin{lemma}\label{lem:cutsizeseq:comp}
	Let $G_1$ and $G_2$ be SPD's. Then the following hold.
	\begin{enumerate}
		\item\label{lem:cutsizeseq:comp:series} $\cutsizeseq(\toporders(G_1 \seriescomp G_2)) = \cutsizeseq(\toporders(G_1)) \allconcat \cutsizeseq(\toporders(G_2))$.
		\item\label{lem:cutsizeseq:comp:parallel} $\cutsizeseq(\toporders(G_1 \parallelcomp G_2)) = \cutsizeseq(\toporders(G_1)) \allmergesnd \cutsizeseq(\toporders(G_2))$.
	\end{enumerate}
\end{lemma}
\begin{proof}
	\cref{lem:cutsizeseq:comp:series}. 
	Let $\cutsizeseq(\linord) \in \cutsizeseq(\toporders(G_1 \seriescomp G_2))$ be such that 
	$\linord$ is a topological order of $G_1 \seriescomp G_2$. 
	Then, $\linord$ consists of two contiguous parts, 
	namely $\linord_1 \defeq \linord|_{V(G_1)} \in \toporders(G_1)$ followed by $\linord_2 \defeq \linord|_{V(G_2)} \in \toporders(G_2)$.
	Since there are no arcs from $V(G_1) \setminus \{t_1\}$ to $V(G_2) \setminus \{s_2\}$,
	we have that 
	$\cutsizeseq(\linord) = \cutsizeseq(\linord_1) \concat \cutsizeseq(\linord_2) 
	\in \cutsizeseq(\toporders(G_1)) \allconcat \cutsizeseq(\toporders(G_2))$. 
	The other inclusion follows similarly.
	
	\cref{lem:cutsizeseq:comp:parallel}. 
	Let $\cutsizeseq(\linord) \in \cutsizeseq(\toporders(G_1 \parallelcomp G_2))$ 
	be such that $\linord$ is a topological order of $G_1 \parallelcomp G_2$. 
	Let $\linord_1 \defeq \linord|_{V(G_1)}$ and $\linord_2 \defeq \linord|_{V(G_2)}$. 
	It is clear that $\linord_1 \in \toporders(G_1)$ and that $\linord_2 \in \toporders(G_2)$.
	Let $\cutsizeseq(\linord) = x(1), \ldots, x(n-1)$, $\cutsizeseq(\linord_1) = y_1(1), \ldots, y_1(n_1 - 1)$,
	and $\cutsizeseq(\linord_2) = y_2(1), \ldots, y_2(n_2-1)$. 
	For any $i \in \{1, \ldots, n-1\}$, 
	let $i_1$ be the maximum index such that $\linord(\linord_1^{-1}(i_1)) \le i$, and define $i_2$ accordingly.
	Then, the set of arcs that cross the cut between positions $i$ and $i+1$ in $\linord$ 
	is the union of the set of arcs crossing 
	the cut between positions $i_1$ and $i_1 + 1$ in $\linord_1$ 
	and the set of arcs crossing the cut between positions $i_2$ and $i_2 + 1$ in $\linord_2$.
	Since $G_1$ and $G_2$ are arc-disjoint, this means that $x(i) = y_1(i_1) + y_2(i_2)$.
	Together with the observation that each vertex at position $i+1 < n$ in $\linord$
	is \emph{either} from $G_1$ \emph{or} from $G_2$, 
	we have that $$x(i+1) \in \{y_1(i_1 + 1) + y_2(i_2), y_1(i_1) + y_2(i_2 + 1)\},$$ 
	in other words, we have that $\cutsizeseq(\linord) \in \cutsizeseq(\linord_1) \allmergesnd \cutsizeseq(\linord_2) \subseteq \cutsizeseq(\toporders(G_1)) \allmergesnd \cutsizeseq(\toporders(G_2))$.
	The other inclusion can be shown similarly, essentially using the fact that we are only considering non-diagonal merges.
\end{proof}

We now prove the crucial lemma of this section which states that we can compute a dominating cut-size sequence of an SPD $G$ from dominating cut-size sequences of the smaller SPD's that $G$ is obtained from. For technical reasons, we assume in the following lemma that $G$ has no parallel arcs, which does not affect the algorithm presented in this section.
\begin{lemma}\label{lem:spd:dominate}
	Let $G$ be an SPD without parallel arcs. Then there is a topological order $\linord^*$ of $G$ such that $\cutsizeseq(\linord^*)$ dominates all cut-size sequences of $G$. Moreover, the following hold. Let $G_1$ and $G_2$ be SPD's and for $r \in [2]$, let $\linord_r^*$ be a topological order of $G_r$ such that $\cutsizeseq(\linord_r^*)$ dominates all cut-size sequences of $G_r$.
	\begin{enumerate}
		\item\label{lem:spd:dominate:series} If $G = G_1 \seriescomp G_2$, then $\linord^* = \linord_1^* \concat \linord_2^*$.
		\item\label{lem:spd:dominate:parallel} If $G = G_1 \parallelcomp G_2$, then $\linord^*$ can be found as the topological order of $G$ such that $\cutsizeseq(\linord^*)$ dominates $\cutsizeseq(\linord^*_1) \allmergesnd \cutsizeseq(\linord_2^*)$.
	\end{enumerate}
\end{lemma}
\begin{proof}
	We prove the lemma by induction on the number of vertices of $G$. If $\card{V(G)} = 2$, then the claim is trivially true (there is only one topological order).
	Suppose that $\card{V(G)} \eqdef n > 2$.
	Since $n > 2$ and $G$ has no parallel arcs, we know that $G$ can be obtained from two SPD's $G_1$ and $G_2$ via series or parallel composition with $\card{V(G_1)} \eqdef n_1 < n$ and $\card{V(G_2)} \eqdef n_2 < n$. By the induction hypothesis, for $r \in [2]$, there is a unique topological order $\linord_r^*$ such that $\cutsizeseq(\linord_r^*)$ dominates all cut-size sequences of $G_r$. 
	
	Suppose $G = G_1 \seriescomp G_2$. 
	Since $\cutsizeseq(\linord_1^*)$ dominates all cut-size sequences of $G_1$ and $\cutsizeseq(\linord_2^*)$ dominates all cut-size sequences of $G_2$, we can conclude using~\cref{lem:BK}\cref{lem:BK:3:19} that $\cutsizeseq(\linord_1^*) \concat \cutsizeseq(\linord_2^*)$ dominates $\cutsizeseq(\toporders(G_1)) \allconcat \cutsizeseq(\toporders(G_2))$ which together with \cref{lem:cutsizeseq:comp}\cref{lem:cutsizeseq:comp:series} allows us to conclude that $\cutsizeseq(\linord_1^*) \concat \cutsizeseq(\linord_2^*) = \cutsizeseq(\linord_1^* \concat \linord_2^*)$ dominates all cut-size sequences of $G$.
	This proves \cref{lem:spd:dominate:series}.
	
	Suppose that $G = G_1 \parallelcomp G_2$, and let $\linord^*$ be a topological order of $G$ such that $\cutsizeseq(\linord^*)$ dominates $\cutsizeseq(\linord_1^*) \allmergesnd \cutsizeseq(\linord_2^*)$. 
	We show that $\cutsizeseq(\linord^*)$ dominates $\cutsizeseq(\toporders(G))$. Let $\linord \in \toporders(G)$. 
	By \cref{lem:cutsizeseq:comp}\cref{lem:cutsizeseq:comp:parallel}, 
	there exist topological orders $\linord_1 \in \toporders(G_1)$ and $\linord_2 \in \toporders(G_2)$ 
	such that $\cutsizeseq(\linord) \in \cutsizeseq(\linord_1) \allmergesnd \cutsizeseq(\linord_2)$. 
	In other words, there are extensions $e_1$ of $\cutsizeseq(\linord_1)$ and $e_2$ of $\cutsizeseq(\linord_2)$ 
	of the same length such that $\cutsizeseq(\linord) = e_1 \merge e_2$. 
	For $r \in [2]$, since $\cutsizeseq(\linord_r^*) \dominates \cutsizeseq(\linord_r)$, 
	we have that $\cutsizeseq(\linord_r^*) \dominates e_r$. 
	By \cref{lem:BK}\cref{lem:BK:3:13},\footnote{Take $r = e_1$, $s = e_2$, $r_0 = \cutsizeseq(\linord_1)$, and $s_0 = \cutsizeseq(\linord_2)$.} 
	there exists some $f \in \cutsizeseq(\linord_1^*) \allmerges \cutsizeseq(\linord_2^*)$
	such that $f \dominates e_1 \merge e_2$,
	and by \cref{lem:merge:nd:equiv}, there is some $f' \in \cutsizeseq(\linord_1^*) \allmergesnd \cutsizeseq(\linord_2^*)$ 
	such that $f' \dominates f$. Since $\cutsizeseq(\linord^*) \dominates \cutsizeseq(\linord_1^*) \allmergesnd \cutsizeseq(\linord_2^*)$,
	we have that $\cutsizeseq(\linord^*) \dominates f'$, and hence \cref{lem:spd:dominate:parallel} follows:
	\begin{align*}
		\cutsizeseq(\linord^*) \dominates f' \dominates f \dominates e_1 \merge e_2 = \cutsizeseq(\linord).
	\end{align*}
\end{proof}

We are now ready to prove the first main result of this section.
\begin{theorem}\label{thm:cutwidth}
	Let $G$ be an SPD on $n$ vertices. There is an algorithm that computes in time $\cO(n^2)$ the cutwidth of $G$, and outputs a topological ordering that achieves the upper bound.
\end{theorem}
\begin{proof}
	We may assume that $G$ has no parallel arcs; if so, we simply subdivide all but one of the parallel arcs. 
	This neither changes the cutwidth, nor the fact that $G$ is series parallel.
	We can therefore apply \cref{lem:spd:dominate} on $G$ in the correctness proof later.
	
	We use the algorithm of Valdes et al.~\cite{VTL82} to compute in time $\cO(n + \card{A(G)})$ a decomposition tree $T$ that yields $G$, see \cref{thm:spd:recognition}. We process $T$ in a bottom-up fashion, and at each node $t \in V(T)$, compute a topological order $\linord_t$ of $G_t$, the series parallel digraph associated with node $t$, such that $\cutsizeseq(\linord_t)$ dominates all cut-size sequences of $G_t$. Let $t \in V(T)$.
	\begin{description}
		\item[Case 1 ($t$ is a leaf node).] In this case, $G_t$ is a single arc and there is precisely one topological order of $G_t$; we return that order.
		\item[Case 2 ($t$ is a series node with left child $\ell$ and right child $r$).] In this case, we look up $\linord_\ell$, a topological order such that $\cutsizeseq(\linord_\ell)$ dominates all cut-size sequences of $G_\ell$, and $\linord_r$, a topological order such that $\cutsizeseq(\linord_r)$ dominates all cut-size sequences of $G_r$. Following \cref{lem:spd:dominate}\cref{lem:spd:dominate:series}, we return $\linord_\ell \concat \linord_r$.
		\item[Case 3 ($t$ is a parallel node with left child $\ell$ and right child $r$).] In this case, we look up $\linord_\ell$ and $\linord_r$ as in Case 2, and we compute $\linord_t$ such that $\cutsizeseq(\linord_t)$ dominates $\cutsizeseq(\linord_\ell) \allmergesnd \cutsizeseq(\linord_r)$ using the Merge Dominator Lemma (\cref{lem:merge:dom:strong}). Following \cref{lem:spd:dominate}\cref{lem:spd:dominate:parallel}, we return $\linord_t$.
	\end{description}
	
	Finally, we return $\linord_{\fr}$, the topological order of $G_\fr = G$, where $\fr$ is the root of $T$.
	\cref{obs:cutsizeseq:dominate,lem:top:order:dominate} ensure that it is sufficient to compute in each of the above cases a set $\toporders_t^* \subseteq \toporders(G_t)$ with the following property. For each $\linord_t \in \toporders(G_t)$, there is a $\linord_t^* \in \toporders_t^*$ such that $\cutsizeseq(\linord_t^*) \dominates \cutsizeseq(\linord_t)$. By \cref{lem:spd:dominate}, we know that we can always find such a set of size one which is precisely what we compute in each of the above cases. Correctness of the algorithm follows.
	Since $T$ has $\cO(n)$ nodes and each of the above cases can be handled in at most $\cO(n)$ time by \cref{lem:merge:dom:strong}, we have that the total runtime of the algorithm is $\cO(n^2)$.
\end{proof}
Our algorithm in fact works for the more general problem of computing the \emph{weighted} cutwidth of a series parallel digraph which we now define formally.
\newcommand\weight{\omega}
\newcommand\weightedcutwidth{\mathsf{w}\cutwidth}
\begin{definition}
	Let $G$ be a directed acyclic graph and $\weight \colon A(G) \to \bN$ be a weight function. For a topological order $\linord \in \toporders(G)$ of $G$, the \emph{weighted cutwidth of $(\linord, \weight)$} is defined as 
	$$\weightedcutwidth(\linord, \weight) \defeq \max\nolimits_{i \in [n-1]}\sum\nolimits_{\stackrel{vw \in A(G)}{\linord(v) \le i, \linord(w) > i}} \weight(vw),$$
	and the \emph{weighted cutwidth of $(G, \weight)$} is $\weightedcutwidth(G, \weight) \defeq \min_{\linord \in \toporders(G)} \weightedcutwidth(\linord, \weight)$.
\end{definition}
The corresponding computational problem is defined as follows.

\fancyproblemdef
	{Weighted Cutwidth of Series Parallel Digraphs}
	{A series parallel digraph $G$ and an arc-weight function $\weight\colon A(G) \to \bN$.}
	{What is the weighted cutwidth of $(G, \weight)$?}
\begin{corollary}\label{cor:weighted:cutwidth}
	Let $G$ be an SPD on $n$ vertices and $\weight\colon A(G) \to \bN$ an arc-weight function. There is an algorithm that computes in time $\cO(n^2)$ the weighted cutwidth of $(G, \weight)$, and outputs a topological ordering that achieves the upper bound.
\end{corollary}

\subsection{Modified Cutwidth}\label{sec:mcutw}
We now show how to use the algorithm for computing the weighted cutwidth of series parallel digraphs from \cref{cor:weighted:cutwidth} to give an algorithm that computes the \emph{modified cutwidth} of a series parallel digraph on $n$ vertices in time $\cO(n^2)$. 
Recall that given a topological order $v_1, \ldots, v_n$ of a directed acyclic graph $G$, 
its modified cutwidth is the maximum over all $i \in [n-1]$ of the number of arcs that have 
their tail vertex in $\{v_1, \ldots, v_{i-1}\}$ and their head vertex in $\{v_{i+1}, \ldots, v_n\}$, 
and that the modified cutwidth of $G$ is the minimum modified cutwidth over all its topological orders.
We are dealing with the following computational problem.

\fancyproblemdef
	{Modified Cutwidth of Series Parallel Digraphs}
	{A series parallel digraph $G$.}
	{What is the modified cutwidth of $G$?}

To solve this problem, we will provide a transformation which allows for applying the algorithm for the
\textsc{Weighted Cutwidth of SPD's} problem to compute the modified cutwidth. 
We would like to remark that this transformation is similar to one provided in~\cite{BFT09}, 
however some modifications are necessary to ensure that the digraph resulting from the transformation is an SPD.
\begin{theorem}\label{thm:mcutw}
	Let $G$ be an SPD on $n$ vertices. There is an algorithm that computes in time $\cO(n^2)$ the modified cutwidth of $G$, and outputs a topological ordering of $G$ that achieves the upper bound.
\end{theorem}
\begin{proof}
	We give a transformation that enables us to solve \textsc{Modified Cutwidth of SPD's} with help of an algorithm that solves \textsc{Weighted Cutwidth of SPD's}.

Let $(G, (s, t))$ be an SPD on $n$ vertices and $m$ arcs. Again, we assume that $G$ has no parallel arcs; if so, we simply subdivide all but one of the parallel arcs. This does not change the (modified) cutwidth, and keeps a digraph series parallel.
We construct another digraph $G'$ and an arc-weight function $\weight\colon A(G') \to \bN$ as follows. For each vertex $v \in V(G) \setminus \{s, t\}$, we add to $G'$ two vertices $v_{in}$ and $v_{out}$. We add $s$ and $t$ to $G'$ and write $s$ as $s_{out}$ and $t$ as $t_{in}$. We add the following arcs to $G'$. First, for each $v \in V(G)$, we add an arc $(v_{in}, v_{out})$ and we let $\weight((v_{in}, v_{out})) \defeq m + 1$. Next, for each arc $(v, w) \in A(G)$, we add an arc $(v_{out}, w_{in})$ to $G'$ and we let $\weight((v_{out}, w_{in})) \defeq 1$.
For an illustration see \cref{fig:reduction}.
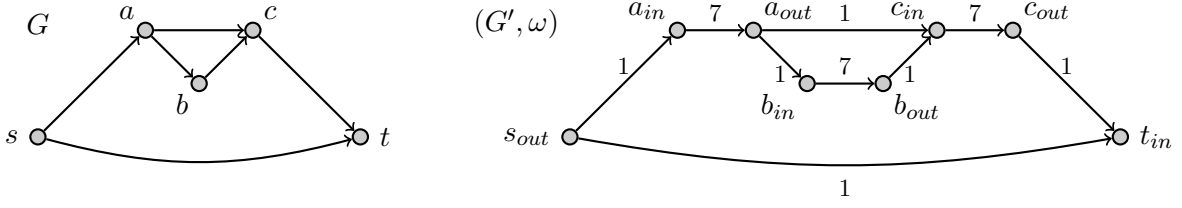
\begin{figure}
	\centering
	\begin{tikzpicture}[auto, node distance = 1cm, thick, 
				vtx/.style={circle,fill=black!20,draw,thick,minimum size=.2cm,inner sep=0pt},
				arc/.style={ ->, thick},
				edgelabel/.style = {scale=0.85, midway}]
		\begin{scope}
			\node[vtx, label=left:{$s$}] (s) at (0, 0) {};
			\node[vtx, above right of=s, node distance=2cm] (a) {};
			\node[anchor = south east] at (a) {$a$};
			\node[vtx, below right of=a] (b) {};
			\node[anchor = north east] at (b) {$b$};
			\node[vtx, above right of=b] (c) {};
			\node[anchor = south west] at (c) {$c$};
			\node[vtx, below right of=c, node distance=2cm, label=right:{$t$}] (t) {};
			
			\node[above of=s, node distance=1.5cm] {$G$};
			
			\draw[arc] (s) to (a);
			\draw[arc] (a) to (b);
			\draw[arc, bend right=15] (s) to (t);
			\draw[arc] (a) to (c);
			\draw[arc] (b) to (c);
			\draw[arc] (c) to (t);
		\end{scope}
		\begin{scope}[xshift=7cm]
			\node[vtx, label=left:{$s_{out}$}] (sOut) at (0, 0) {};
			\node[vtx, above right of=sOut, node distance=2cm] (aIn) {};
			\node[anchor = south east] at (aIn) {$a_{in}$};
			\node[vtx, right of=aIn] (aOut) {};
			\node[anchor = south west] at (aOut) {$a_{out}$};
			\node[vtx, below right of=aOut] (bIn) {};
			\node[anchor = north east] at (bIn) {$b_{in}$};
			\node[vtx, right of=bIn] (bOut) {};
			\node[anchor = north west] at (bOut) {$b_{out}$};
			\node[vtx, above right of=bOut] (cIn) {};
			\node[anchor = south east] at (cIn) {$c_{in}$};
			\node[vtx, right of=cIn] (cOut) {};
			\node[anchor = south west] at (cOut) {$c_{out}$};
			\node[vtx, below right of=cOut, node distance=2cm, label=right:{$t_{in}$}] (tIn) {};
			
			\node[above of=sOut, node distance=1.5cm, anchor=east] {$(G', \omega)$};
			
			\draw[arc] (aIn) -> (aOut) node [edgelabel, above] {$7$};
			\draw[arc] (bIn) -> (bOut) node [edgelabel, above] {$7$};
			\draw[arc] (cIn) -> (cOut) node [edgelabel, above] {$7$};
			\draw[arc, bend right=10] (sOut) to (tIn);
			\node[anchor = north] at ($(sOut)!0.5!(tIn)$) [below=0.45cm, scale=0.85] {$1$};
			\draw[arc] (sOut) -> (aIn) node [edgelabel, above] {$1$};
			\draw[arc] (aOut) -> (bIn) node [edgelabel, below] {$1$};
			\draw[arc] (bOut) -> (cIn) node [edgelabel, below] {$1$};
			\draw[arc] (cOut) -> (tIn) node [edgelabel, above] {$1$};
			\draw[arc] (aOut) -> (cIn) node [edgelabel, above] {$1$};
		\end{scope}
	\end{tikzpicture}
	\caption{Illustration of the transformation given in the proof of \cref{thm:mcutw}.
		Note that in this case, $m =6$, so the arcs between vertices $v_{in}$ and $v_{out}$ have weight $7$.}
	\label{fig:reduction}
\end{figure}

We observe that the size of $G'$ is linear in the size of $G$, and then prove that if $G'$ is obtained from applying the above transformation to a series parallel digraph, then $G'$ is itself an SPD.
\begin{nestedobservation}\label{obs:mcw:size}
	Let $G$ and $G'$ be as above. Then, $n' \defeq \card{V(G')} \le 2\card{V(G)}$ and $\card{A(G')} \le \card{A(G)} + \card{V(G)}$.
\end{nestedobservation}
\begin{nestedclaim}\label{claim:mcw:spd}
	If $G$ is a series parallel digraph, then $G'$ as constructed above is an SPD.
\end{nestedclaim}
\begin{claimproof}
	We prove the claim by induction on $n$, the number of vertices of $G$. For the base case when $n = 2$, we have that $G$ is a single arc in which case $G'$ is a single arc as well. 
	Now suppose $n > 2$.
	Since $n > 2$, $G$ is obtained from two series parallel digraphs $G_1$ and $G_2$ via series or parallel composition. Since $G$ has no parallel arcs, we can use the induction hypothesis to conclude that the graphs $G_1'$ and $G_2'$ obtained via our construction are series parallel. Now, if $G = G_1 \parallelcomp G_2$, then it is immediate that $G'$ is series parallel. If $G = G_1 \seriescomp G_2$, then we have that in $G'$, the vertex that was constructed since $t_1$ and $s_2$ were identified, call this vertex $x$, got split into two vertices $x_{in}$ and $x_{out}$ with a directed arc of weight $m+1$ pointing from $x_{in}$ to $x_{out}$. Call the series parallel digraph consisting only of this arc $(X, (x_{in}, x_{out}))$. We now have that $G' = G_1' \seriescomp X \seriescomp G_2'$, so $G'$ is series parallel in this case as well.
\end{claimproof}

We are now ready to prove the correctness of this transformation. To do so, we will assume that we are given an integer $k$ and we want to decide whether the modified cutwidth of $G$ is at most $k$.
\begin{nestedclaim}\label{claim:mcw:if}
	If $G$ has modified cutwidth at most $k$, then $G'$ has weighted cutwidth at most $m + k + 1$.
\end{nestedclaim}
\begin{claimproof}
	Take a topological ordering $\linord$ of $G$ such that $\modifiedcutwidth(\linord) \le k$. We obtain $\linord'$ from $\linord$ by replacing each vertex $v \in V(G) \setminus \{s, t\}$ by $v_{in}$ followed directly by $v_{out}$. Clearly, this is a topological order of $G'$. We show that the weighted cutwidth of this ordering is at most $m + k + 1$. 
	
	Let $i \in [n'-1]$ and consider the cut between position $i$ and $i+1$ in $\linord'$. We have to consider two cases. In the first case, there is some $v \in V(G)$ such that $\linord'^{-1}(i) = v_{in}$ and $\linord'^{-1}(i+1) = v_{out}$. Then, there is an arc of weight $m+1$ from $v_{in}$ to $v_{out}$ crossing this cut, and some other arcs of the form $(u_{out}, w_{in})$ for some arc $(u, w) \in A(G)$. All these arcs cross position $\linord(v)$ in $\linord$, so since $\modifiedcutwidth(\linord) \le k$, there are at most $k$ of them. Furthermore, for each such arc we have that $\weight((u_{out}, w_{in})) = 1$ by construction, so the total weight of this cut is at most $m + k + 1$.
	
	In the second case, we have that $\linord'^{-1}(i) = v_{out}$ and $\linord'^{-1}(i+1) = w_{in}$ for some $v, w \in V(G)$, $v \neq w$. By construction, we have that $\linord(w) = \linord(v) + 1$. Hence, any arc crossing the cut between $i$ and $i+1$ in $\linord'$ is of one of the following forms.
	\begin{enumerate}
		\item It is $(x_{out}, y_{in})$ for some $(x, y) \in A(G)$ with $\linord(x) < \linord(v)$ and $\linord(y) > \linord(v)$, or
		\item it is $(x_{out}, y_{in})$ for some $(x, y) \in A(G)$ with $\linord(x) < \linord(w)$ and $\linord(y) > \linord(w)$, or
		\item it is $(v_{out}, w_{in})$.
	\end{enumerate}
	Since $\modifiedcutwidth(G) \le k$, there are at most $k$ arcs of the first and second type, and since $G$ has no parallel arcs, there is at most one arc of the third type. By construction, all these arcs have weight one, so the total weight of this cut is $2k + 1 \le m + k + 1$.
\end{claimproof}

\begin{nestedclaim}\label{claim:mcw:onlyif}
	If $G'$ has weighted cutwidth at most $m + k + 1$, then $G$ has modified cutwidth at most $k$.
\end{nestedclaim}
\begin{claimproof}
	Let $\linord'$ be a topological order of $G'$ such that $\weightedcutwidth(\linord', \weight) \le m + k + 1$. First, we claim that for all $v \in V(G) \setminus \{s, t\}$, we have that $\linord'(v_{out}) = \linord'(v_{in}) + 1$. Suppose not, for some vertex $v$. If we have that $\linord'(v_{in}) < \linord'(w_{in}) < \linord'(v_{out})$ for some $w \in V(G) \setminus \{s, t\}$ and $w \neq v$, then the cut between $\linord'(w_{in})$ and $\linord'(w_{in}) + 1$ has weight at least $2m + 2$: the two arcs $(v_{in}, v_{out})$ and $(w_{in}, w_{out})$ cross this cut, and they are of weight $m + 1$ each. Similarly, if $\linord'(v_{in}) < \linord'(w_{out}) < \linord'(v_{out})$, then the cut between $\linord'(w_{out}) - 1$ and $\linord'(w_{out})$ has weight at least $2m + 2$. Since $2m + 2 > m + k + 1$, we have a contradiction in both cases.
	
	We define a linear ordering $\linord$ of $G$ as follows. We let $\linord(s) \defeq 1$, $\linord(t) \defeq n$, and for all $v, w \in V(G) \setminus \{s, t\}$, we have $\linord(v) < \linord(w)$ if and only if $\linord'(v_{in}) < \linord'(w_{in})$. It is clear that $\linord$ is a topological ordering of $G$; we show that $\linord$ has modified cutwidth at most $k$. Consider an arc $(x, y)$ that crosses a vertex $v$ in $\linord$, i.e.\ we have that $\linord(x) < \linord(v) < \linord(y)$. We have just argued that $\linord'(v_{out}) = \linord'(v_{in}) + 1$, so we have that the arc $(x_{out}, y_{in})$ crosses the cut between $v_{in}$ and $v_{out}$ in $\linord'$. Recall that there is an arc of weight $m + 1$ from $v_{in}$ to $v_{out}$, so since $\weightedcutwidth(\linord', \weight) \le m + k + 1$, we can conclude that in $\linord$, there are at most $(m + k + 1) - (m - 1) = k$ arcs crossing the vertex $v$ in $\linord$.
\end{claimproof}
Now, to compute the modified cutwidth of $G$, we run the above described transformation to obtain $(G', \weight)$, and compute a topological order that gives the smallest weighted cutwidth of $(G', \weight)$ using \cref{cor:weighted:cutwidth}. We can then follow the argument given in the proof of \cref{claim:mcw:onlyif} to obtain a topological order for $G$ that gives the smalles modified cutwidth of $G$.

By \cref{claim:mcw:spd}, $G'$ is an SPD, so we can indeed apply the algorithm of \cref{cor:weighted:cutwidth} to solve the instance $(G', \weight)$. Correctness follows from \cref{claim:mcw:if,claim:mcw:onlyif}. By \cref{obs:mcw:size}, $\card{V(G')} = \cO(\card{V(G)}) = \cO(n)$, and clearly, $(G', \weight)$ can be constructed in time $\cO(\card{V(G)} + \card{A(G)})$; so the overall runtime of this procedure is at most $\cO(n^2)$.
\end{proof}

\section{Conclusions}
\label{section:conclusions}
In this paper, we obtained a new technical insight in a now over a quarter century old technique, namely the use of typical sequences. 
The insight led to new polynomial time algorithms. 
Since its inception, algorithms based on typical sequences give the best
asymptotic bounds for linear time FPT algorithms for treewidth and pathwidth, as
functions of the target parameter. It still remains a challenge to improve
upon these bounds ($2^{O(pw^2)}$, respectively $2^{O(tw^3)}$), or give
non-trivial lower bounds for parameterized pathwidth or treewidth. Possibly, the Merge Dominator Lemma can be helpful to get some progress here.

As other open problems, we ask whether there are other width parameters
for which the Merge Dominator Lemma implies polynomial time or XP algorithms,
or whether such algorithms exist for other classes of graphs. 
For instance, 
for which width measures can we give XP algorithms when parameterized by
the treewidth of the input graph?

\bibliographystyle{plain}
\bibliography{ref}

\end{document}